\documentclass[pra,twocolumn]{revtex4}

\usepackage{lmodern}
\usepackage[utf8]{inputenc}
\usepackage[T1]{fontenc}

\usepackage{graphicx}
\usepackage[caption=false]{subfig} %old, switch to subcaption at some point
\graphicspath{{./Figures/}}

\usepackage{amsmath}
\usepackage{amssymb,amsthm}
\usepackage{bm}
\usepackage{bbm}

\newtheorem{theorem}{Theorem}
\newtheorem{lemma}{Lemma}

\newcounter{mycount}

\newcommand{\be}[1]{ \begin{equation} \begin{aligned}{#1} }    
\newcommand{\ee}{\end{aligned} \end{equation}}

\usepackage{braket}
\usepackage{multirow,bigdelim}

\begin{document}

\title{Composite fermion basis for two-component Bose gases}
\author{O. Liab\o tr\o}
\author{M. L. Meyer}
\affiliation{Department of Physics, University of Oslo, P.O. Box 1048 Blindern,
0316 Oslo, Norway}

\date{\today}

\begin{abstract}
Despite its success, the composite fermion (CF) construction possesses some
mathematical features that have, until recently, not been fully understood. In
particular, it is known to produce wave functions that
are not necessarily orthogonal, or even linearly independent, after projection to the lowest Landau level.
While this is usually not a problem in practice in the quantum Hall regime, we have
previously shown that it presents a technical challenge for rotating Bose gases
with low angular momentum. These are systems where the CF approach yields
surprisingly good approximations to the exact eigenstates of weak short-range
interactions, and so solving the problem of linearly dependent wave functions is
of interest. It can also be useful for studying higher bands of fermionic quantum Hall states.
Here we present several ways of constructing a basis for the space of so-called ``simple''
CF states for two-component rotating Bose gases in the lowest Landau level,
and prove that they all give sets of linearly independent wave functions that span the space. Using this basis, we study the structure of
the lowest-lying state using so-called restricted wave functions. We also
examine the scaling of the overlap between the exact and CF wave functions at
the maximal possible angular momentum for simple states.

\end{abstract}
\pacs{ }

\maketitle

%%%%%%%%%%%%%%%%
%%%%%%%%%%%%%%%%

\section{Introduction}
\label{sec:intro}
Almost 20 years ago, the connection between the physics of charged particles
moving in two dimensions in a strong magnetic field, and dilute cold atoms
rotating rapidly in a harmonic trap, was noticed \cite{wilkin-gunn-smith98}.
Since then, a large body of theoretical and experimental work has accumulated
that explore the various aspects of rapidly rotating atomic gases and the
associated quantum phenomena; for reviews, see e.g. \cite{viefers-review, cooper-review, fetter-review}. In particular, one
expects strongly correlated phases similar to those found in the quantum Hall
effect when the atoms experience strong synthetic magnetic fields. The effect of
a synthetic magnetic field can be generated by simply rotating the cloud, or by
more advanced techniques \cite{cornell04, roncaglia11, gemelke10, lin09, dalibardreview}.

For electron systems, one prominent way of theoretically studying the quantum
Hall effect involves constructing (classes of) explicit trial wave functions
that approximate the true low-energy eigenstates of the interacting system. At
least in the case of Coulomb interaction, the true many-body eigenstates are
extremely complicated. Still, many successful trial wave functions exist, most
famously the Laughlin wave functions \cite{laughlin83}, the family of composite fermion (CF)
states \cite{jain89}, and trial wave functions addressing non-Abelian quantum Hall
states \cite{moore91, read99, ardonne-schoutens99}. The success of these trial wave functions in explaining various
phenomena is linked to the way they capture the important topological
properties of the phases they describe.

Many of the methods mentioned above have been modified to be applicable to
weakly interacting cold atom systems \cite{cooper-wilkin99, viefers00}. The hope is to be able to experimentally
study strongly correlated states in a cold atom setting, where parameters like
density, disorder and scattering lengths may be tuned much more finely than in the
semiconductor systems traditionally used in quantum Hall experiments \cite{bloch08}. More
recently, an additional degree of freedom is often taken into consideration in models and experiments, called
pseudospin \cite{hall98, kasamatsu03}. Typically this means multicomponent mixtures where the different
components are different internal states of the atoms. Varying inter- and
intraspecies interactions independently allows for novel behaviours \cite{kasamatsu-tsubota-ueda05}.
This pseudospin degree of freedom has
been incorporated in the composite fermion scheme used for cold atom systems,
and has been used in the quantum Hall regime of high angular momenta \cite{jain13, grass14}. On the
other hand, near the lower end of the angular momentum scale, where the CF
description is not {\it a priori} expected to work, it has been shown that it
actually works surprisingly well, both in scalar and two-component cases \cite{korslund06, meyer14}.

A property of the CF method of constructing wave functions is that one typically
needs to do a projection into the lowest Landau level (LLL) in order to either
compare different CF wave functions, or to compare a CF to an eigenstate found
by numerical diagonalization of the interacting Hamiltonian \cite{jainbook}. This projection
leads to non-orthogonal, and often also linearly dependent, CF wave functions.
This has been known in the context of electrons in the quantum Hall regime \cite{dev-jain92, wu95, balram13}, but
the issue is much more prominent for bosons with low angular momentum: we have previously observed \cite{meyer-lia16} one or two orders of magnitude difference between the number of seemingly distinct CF
candidates and the actual number of linearly independent wave functions. In some
extreme cases the number of seemingly distinct CF candidates one can write down
is even larger than the dimension of the relevant sector of Hilbert space, meaning they cannot possibly be independent. Until
very recently, little was understood about the mechanisms responsible for these
linear dependencies after projection. In a previous paper \cite{meyer-lia16}, we discussed three
types of relations between certain types of low-lying CF candidates, but were
only able to give examples demonstrating how these relations seem to explain all
the linear dependencies.

In this paper, we present sets of CF candidates for two-component systems that we rigorously prove are
basis sets for the subspace that minimizes the CF
cyclotron energy for low angular momenta. We use these states to study the real-space distribution of particles and vortices of the lowest-lying
wave functions when we vary the angular momentum, and also give additional
attention to certain special cases, comparing to known behaviour from scalar
gases.

%%%%%%%%%%%%%%%%
%%%%%%%%%%%%%%%%

\section{Two-component rotating Bose gases}
\label{sec:theory}

We first summarize the model for two-species Bose gases in the lowest Landau
level, including their description in terms of composite fermions. A more
detailed introduction can be found in \cite{meyer14}. The two species of bosons
experience a two-dimensional harmonic trap potential of strength $\omega$, and
are rotating about the minimum of the potential at frequency $\Omega$. The
Hamiltonian is:
\be
H & = \sum_{i=1}^{N+M} \left(
\frac{\mathbf{p}_i^2}{2m}+\frac{1}{2}m\omega^2\mathbf{r}_i^2 - \Omega l_i
\right) \\
  & \ + \sum_{i=1}^{N+M} \sum_{j=i+1}^{N+M} 2\pi g_{i,j} \delta(\mathbf{r}_i - \mathbf{r}_j).  
\ee
Here $N$ denotes the number of particles of the minority species, and $M \geq N$
denotes the number of particles of the majority species, all with the same mass
$m$. $l_i$ is the angular momentum of particle $i$. The strength of the contact
interaction $g_{i,j}$ depends only on the species of particles $i, j$. In the
species-independent case $g_{i,j} = g =$ constant, the system posesses a
pseudospin-1/2 symmetry, which we will assume here. For sufficiently dilute
gases, i.e. in the weak interaction limit, this reduces to the well known lowest
Landau level problem \cite{viefers-review, cooper-review} in the effective magnetic field $2m\omega$,
\begin{align}
H &= \sum_{i}(\omega-\Omega)l_i + 2\pi g\sum_{i<j} \delta(\eta_i
- \eta_j).
\end{align}
In the ideal limit $(\omega - \Omega) \rightarrow 0$ the Landau levels become
flat, meaning that the many-body eigenstates are solely determined by the
interaction.
Here $\eta_j = x_j + i y_j$ are the dimensionless complex positions of the
particles in units of the ``magnetic'' length $\sqrt{\hbar/(2m\omega)}$. We name
the coordinates of the two components $z_i = \eta_i$, $1 \leq i \leq N$ and $w_i
= \eta_{N + i}$, $1 \leq i \leq M$.
Working in symmetric gauge, the single particle eigenstates in the lowest Landau
level with angular momentum $l$ are 
\begin{equation}
\psi_{0,l}(\eta) = N_{l} \eta^l \exp{(-\eta\bar{\eta}/4)} \qquad l\geq 0
\label{spwf}
\end{equation}
The Gaussian factors are ubiquitous, so we suppress them for simplicity from now
on. Since the Hamiltonian commutes with the total angular momentum $L = \sum_i
l_i$, we may work with many-body wave functions that are eigenstates of $L$.
These are homogeneous polynomials of degree $L$, symmetric in the coordinates of
each species separately. As is common \cite{viefers-review, meyer14} we will focus on
translationally invariant states, i.e. polynomials invariant under a constant
shift $\eta'$ of all coordinates,
\begin{equation}
\Psi(\{\eta_i + \eta' \}) = \Psi(\{\eta_i \}).
\end{equation}

As mentioned in the introduction, one may adapt the CF approach to produce wave
functions for 2D bosons, both in the scalar and multi-component cases. A CF
trial wave function for the bosonic two-species system is of the
form \cite{jainbook}
\begin{equation}
 \Psi_{\text{CF}}=\mathcal{P}_{\text{LLL}} \left(\Phi_Z \Phi_W J(z,w)^{q} \right)
 \label{2swf}
\end{equation}
where $q$ is an \emph{odd} number, $q=1,3,5,\ldots$. $\Phi_Z$, $\Phi_W$ are
Slater determinants for each species of CFs, treated as non-interacting to
lowest order. They consist of CF orbitals
\begin{equation}
\psi_{n,m}(\eta) = N_{n,m} \eta^m L_n^m\left( \frac{\eta \bar{\eta }}{2}
\right),
\quad m \geq -n 
\label{spwf2},
\end{equation}
where $L_n^m$ is the associated Laguerre polynomial, and $N_{n,m}$ is a
normalization factor. The interpretation of these orbitals is that they are
composite fermions occupying Landau-like levels, labeled by $n$ (often called
$\Lambda$-levels) in a reduced effective magnetic field.
$J$ is a Jastrow factor involving both species,
\be{}
J(z,w) &= \prod_{i<j}(\eta_i - \eta_j) \\	
 &= \prod_{i<j} (z_i-z_j) \prod_{k<l} (w_k-w_l) \prod_{i,k}(z_i-w_k).
\ee
$J(z,w)^q$ has $q$ units of angular momentum per pair of particles, in total
$L_J = q(N+M)(N+M-1)/2$. We will be considering low total angular momentum,
hence we will choose the smallest possibility $q=1$. $\mathcal{P}_{\text{LLL}} $
denotes projection to the lowest Landau level. We use the projection of Girvin
and Jach \cite{girvin-jach84} (called Method I in \cite{jainbook}), which amounts to
first moving the conjugate variables $\overline{\eta}_i$ to the left of the
$\eta_i$, and then replacing $\overline{\eta}_i$ by $\partial_{\eta_i}$.

In this paper, we focus on low angular momenta, specifically $L \leq M N$. We consider the set of translationally invariant CF wave
functions that minimize the total CF cyclotron energy $E_c \propto \sum_i n_i$ in this $L$ range. The sum runs over the CF orbitals occupied in a pair of Slater determinants. Since $L \leq M N < 1/2(N+M)(N+M-1) = L_J$, we see that $\sum_i m_i$ must be \emph{negative}. We have previously shown that this set is spanned by CF candidates where only the orbitals $\psi_{n, -n}$ are occupied, and that linear combinations of candidates in this set give good overlaps with the very lowest-lying states in the exact yrast spectrum \cite{meyer14}. After projection to the LLL, the orbitals $\psi_{n,-n}(\eta)$ become $\partial_{\eta}^n$. This simple form of the Slater determinants has led these CF wave functions to be called {\it simple states}. Since the differential operators commute, we may perform the projection to the LLL on the Slater determinants individually, as long as we keep them to the left of the Jastrow factor. This is understood in the following sections.

%%%%%%%%%%%%%%%%%
%%%%%%%%%%%%%%%%%

\section{Composite fermion basis}
\label{sec:basis}

We now present some sets of states that we prove to be bases for the space of simple states with $N+M$ particles and angular momentum $L$.  First, we need some definitions.  The simple states are on the form (\ref{2swf}), where the Slater determinants after projection are
\be{}
\Phi_Z(\mathbf{a})=\sum_{\rho\in S_N}\prod_{i=1}^N(-1)^{|\rho|}\partial_{z_{\rho_i}}^{a_i}
\ee
and 
\be{}
\Phi_W(\mathbf{b})=\sum_{\rho\in S_M}\prod_{i=1}^M(-1)^{|\rho|}\partial_{w_{\rho_i}}^{b_i}
\ee
with $a_i, b_i<N+M-1$.  In the following, we will focus on $\Phi_Z$, but we will never use that $N\leq M$, so replacing $W, w, M\leftrightarrow Z, z, N$ is always possible.

We define
\be{}
P_{N,M,L}=\{p\in\mathbb{Z}^N\ |\ M\geq p_1\geq...\geq p_N\geq0,\sum_{i=1}^N p_i=L \}.
\ee
This is the set of partitions of $L$ into an $M\times N$ box.  We order the partitions lexicographically.  The partitions can be visualized using Young diagrams.  Such a diagram is obtained by coloring $L$ cells of an $M\times N$ box compactly from the lower left. The number of colored cells in the lowest row is $p_1$, the next row corresponds to $p_2$, and so on.  As an example, we show all the Young diagrams corresponding to the set $P_{3,5,6}$ in FIG. \ref{fig:youngdiagrams}.  

\begin{figure}
\includegraphics[width=\columnwidth]{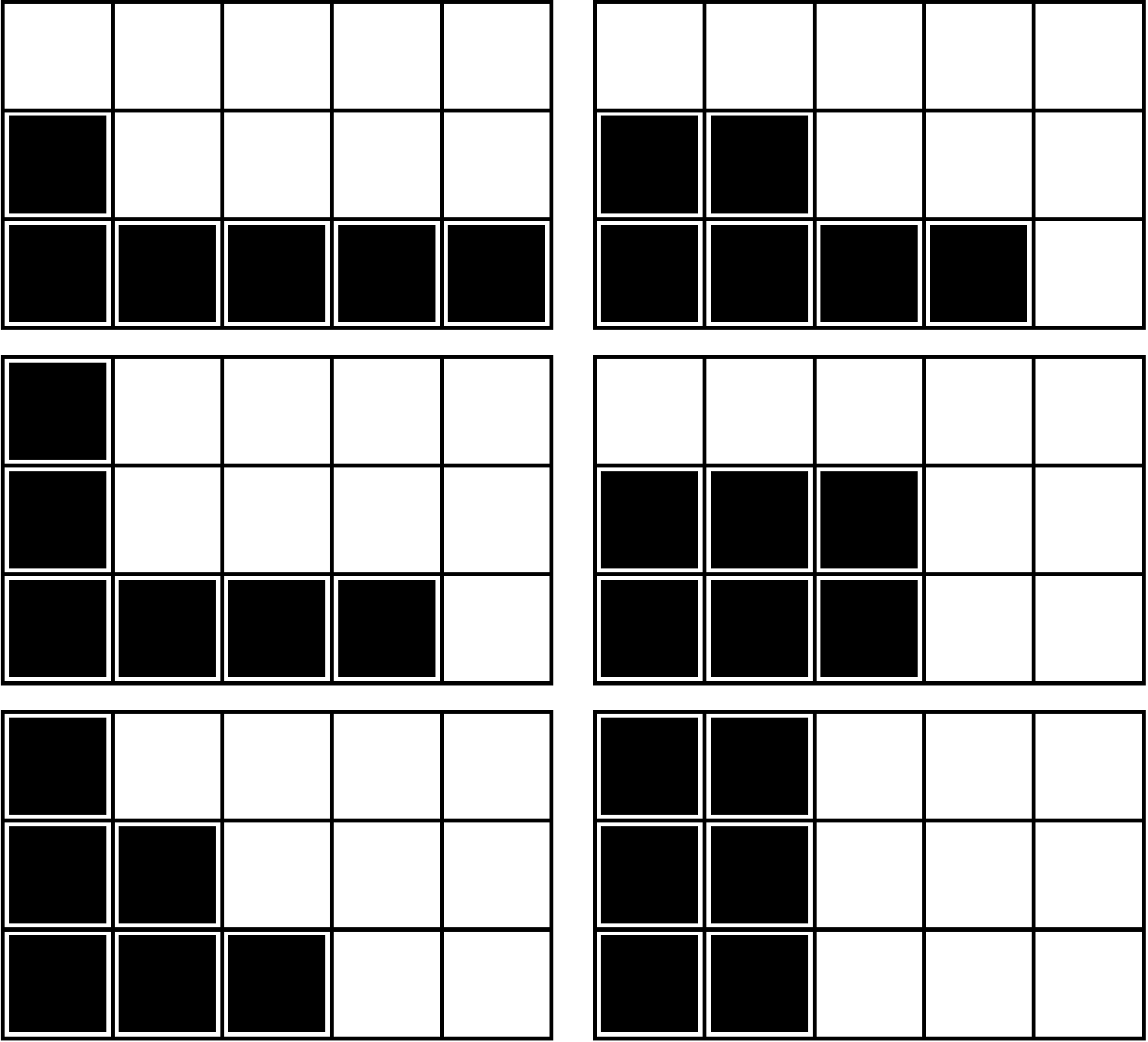}
\caption{Young diagrams corresponding to partitions of 6 into a $5\times 3$ box. \label{fig:youngdiagrams}}
\end{figure}

There is a one to one correspondence between $P_{N,M,\Lambda}$ and $\{\Phi_Z(\mathbf{a})|\sum_{i=1}^N a_i+1-i=\Lambda\}$ given by 
\be{}
p\leftrightarrow \Phi_Z(\mathbf{a}(p)),\ \ a_i(p)=i-1+p_{N+1-i}.
\ee
We define the following differentiation operators:

\be{}
\Delta_{z^n}\equiv\sum_{i=1}^N \partial_{z_i}^n,\ \Delta_{w^n}\equiv\sum_{i=1}^{M} \partial_{w_i}^n,
\ee{}
and
\be{}
\Delta_Z(p\in P_{N,M,\Lambda})\equiv\sum_{\rho\in S_N}\prod_{i=1}^N\partial_{z_{\rho_i}}^{p_i}. 
\ee{}
We have that 
\be{}\label{symdiff}
\Delta_Z(p)\Phi_Z(\bm{a}) &=\sum_{\sigma,\rho\in S_N}(-1)^{|\rho|}\prod_{i=1}^N\prod_{j=1}^N\partial_{z_{\sigma_i}}^{p_i}\partial_{z_{\rho_i}}^{a_i}\\
&=\sum_{\sigma,\rho\in S_N}(-1)^{|\rho|}\prod_{i=1}^N\prod_{j=1}^N\partial_{z_{\rho_i}}^{p_{\sigma_i}}\partial_{z_{\rho_i}}^{a_i}\\
&=\sum_{\sigma\in S_N} \Phi_Z(\bm{a}+\sum_{k=1}^N p_{\sigma_k}\mathbf{e}_k)
\ee
and
\be{}\label{easysymdiff}
\Delta_{z^n}\Phi_Z(\bm{a})=\sum_{k=1}^N\Phi_Z(\bm{a}+n\mathbf{e}_k).
\ee
The maximal $L=MN$ state occurs when the Slater determinants contain minimal differentiation operators, i.e. 
\be{}
\Psi_{(L=MN)}=\Phi_Z(\bm{\alpha})\Phi_W(\bm\beta)J, \label{eq:Lmax}
\ee
where $\alpha_i, \beta_i=i-1$.  We will continue to use the vectors $\bm\alpha$ and $\bm\beta$ as they are useful not only for $L=MN$ but also for general $L$.  We can now state our main theorem:\\

\begin{theorem}
The sets
\be{}
B_{Z,L} &=\{\Phi_Z(\bm a(p))\Phi_W(\bm\beta)J\ |\ p\in P_{N,M,MN-L}\}\\
D_{Z,L} &=\{\Delta_Z(p)\Phi_Z(\bm{\alpha})\Phi_W(\bm\beta\  )J\ |\ p\in P_{N,M,MN-L}\}
\label{eq:basissets}
\ee
are bases for the set of simple CF states with $N+M$ particles and angular momentum $L$.  The sets $B_{W,L}$ and $D_{W,L}$ are also bases.
\end{theorem}

In fact, $B_{Z,L}=B_{W,L}$ for $MN-L$ even.  Otherwise, the sets contain $(-1)$ times the vectors of the other.  This can be seen as a consequence of reflection symmetry, as introduced in \cite{meyer-lia16}.

\begin{proof}
We first show that $D_{Z,L}$ spans $B_{Z,L}$; we already know from Eq. (\ref{symdiff}) that $B_{Z,L}$ spans $D_{Z,L}$.  We then show that $B_{Z,L}$ spans the set of simple CF states, and finally that $B_{Z,L}$ is a linearly independent set.  This must also hold for $D_{Z,L}$ since $|D_{Z,L}|=|B_{Z,L}|$.  We refer to lemmas that we prove in the appendix.\\

Lemma \ref{Lemma1} states that $D_{Z,L}$ spans $B_{Z,L}$:
\be{}
\Phi_Z(\mathbf{a})=\sum_{p\in P_{N,M,\Lambda}} c_p\Delta_Z(p)\Phi_Z(\bm{\alpha})
\ee
for coefficients $\bm c\in \mathbb{Q}^{|P_{N,M,\Lambda}|}$, where $\Lambda=MN-L$. The lemma also shows that we can write a general simple CF state as
\be{}
\Psi=\phi_Z(\bm a)\phi_W(\bm b)J=\sum_{p\in P_{N,M,\Lambda}} c_p\Delta_Z(p)\phi_Z(\bm \alpha)\phi_W(\bm b)J.
\ee
Lemma \ref{Lemma2} states that we can write
\be{}\label{factorlemma}
\Delta_Z(p) =\sum_{\tilde{p}\in P_{N,\Lambda,\Lambda}}d_{p\tilde{p}}\prod_{i=1}^N \Delta_{z^{\tilde{p}_i}},
\ee
with coefficients $d_{p\tilde p} \in \mathbb{Q}$. We can therefore write 
\be{}
\Psi=\sum_{p\in P_{N,M,\Lambda}}\sum_{\tilde p\in P_{N,\Lambda,\Lambda}} c_pd_{p\tilde p}\prod_{i=1}^N\Delta_{z^{\tilde p_i}}\phi_Z(\bm \alpha)\phi_W(\bm b)J.
\ee
Next, we can apply generalized translation invariance (Lemma \ref{gentrans}),
\be{}
(\Delta_{z^n}+\Delta_{w^n})\Psi=0, 
\ee
to replace the $\Delta_{z^{\tilde{p}_i}}$ operators with $-\Delta_{w^{\tilde{p}_i}}$, giving

\be{}
\Psi=\sum_{p\in P_{N,M,\Lambda}}\sum_{\tilde p\in P_{N,\Lambda,\Lambda}} s(\tilde p)c_pd_{p\tilde p}\prod_{i=1}^N\Delta_{w^{\tilde p_i}}\phi_Z(\bm \alpha)\phi_W(\bm b)J,
\ee
where $s(\tilde p)$ is $(-1)$ if there is an odd number of $\tilde p_i$ that are non-zero, and $1$ otherwise.

We can see from eqs. (\ref{symdiff}, \ref{easysymdiff}) that this gives terms that are all on the form $\Phi_Z(\bm\alpha)\Phi_Z(\bm b') J$ for varying $\bm b'$.  We have now expressed a general simple state as a linear combination of elements of $B_{W,L}$, and the same can be done for $B_{Z,L}$.  All that remains is to show that $B_{Z,L}$ is a linearly independent set, and this result is Lemma \ref{linearindependence}.
\end{proof}

\begin{figure}
\includegraphics[width=\columnwidth]{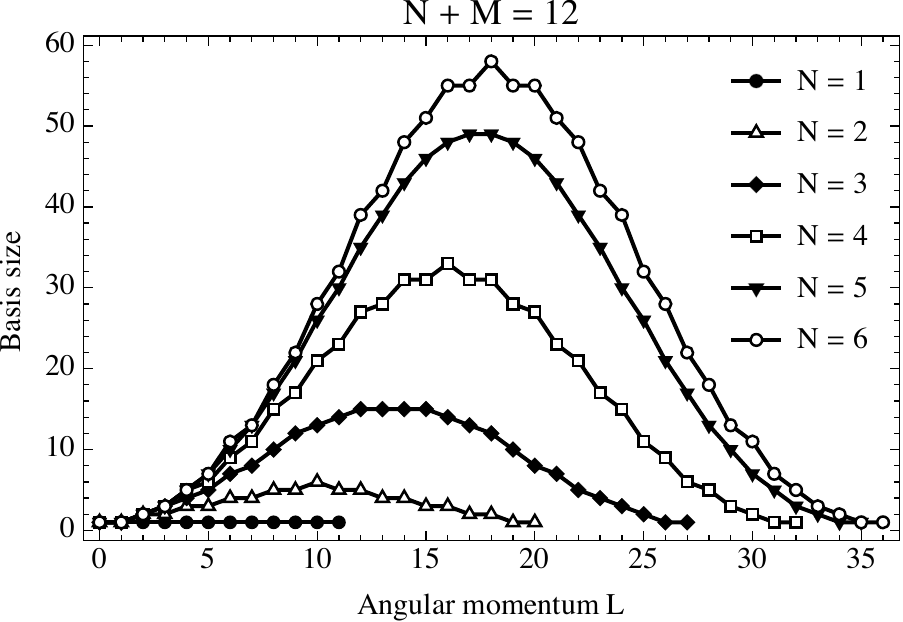}
\caption{The size of the simple CF basis as a function of $L$ for a total of $N + M = 12$ particles.
\label{fig:basissize}}
\end{figure}

To summarize, we have now shown that either of the sets $B_{Z, L}, B_{W, L}, D_{Z, L}, D_{W, L}$ (Eq. (\ref{eq:basissets}) ) are
basis sets for the simple CF wave functions at angular momentum $L$. In other words, when constructing basis sets for simple states, we can \emph{always} fix $\Phi_W$ ($\Phi_Z$) to $\Phi_W(\bm\beta)$ ($\Phi_Z(\bm\alpha)$) and then vary the $Z$ ($W$) determinant. The number of ways to do that, i.e. the \emph{size} of the basis, is simply $|P_{N, M, M N - L}|$. We plot this as function of $L$ for $N + M = 12$ in FIG. \ref{fig:basissize}. The symmetry of each curve about its midpoint $L=M N /2$ that was originally observed in \cite{meyer-lia16} is now easily understood because the number of ways to color $L$ cells in a Young diagram leaves $M N - L$ cells uncolored, meaning that the $L$ diagrams are 1-1 with the $M N - L$ diagrams.  We have illustrated this in FIG. \ref{fig:reflectionrotation}.

\begin{figure}
\subfloat[]{
\includegraphics[width=0.45\columnwidth]{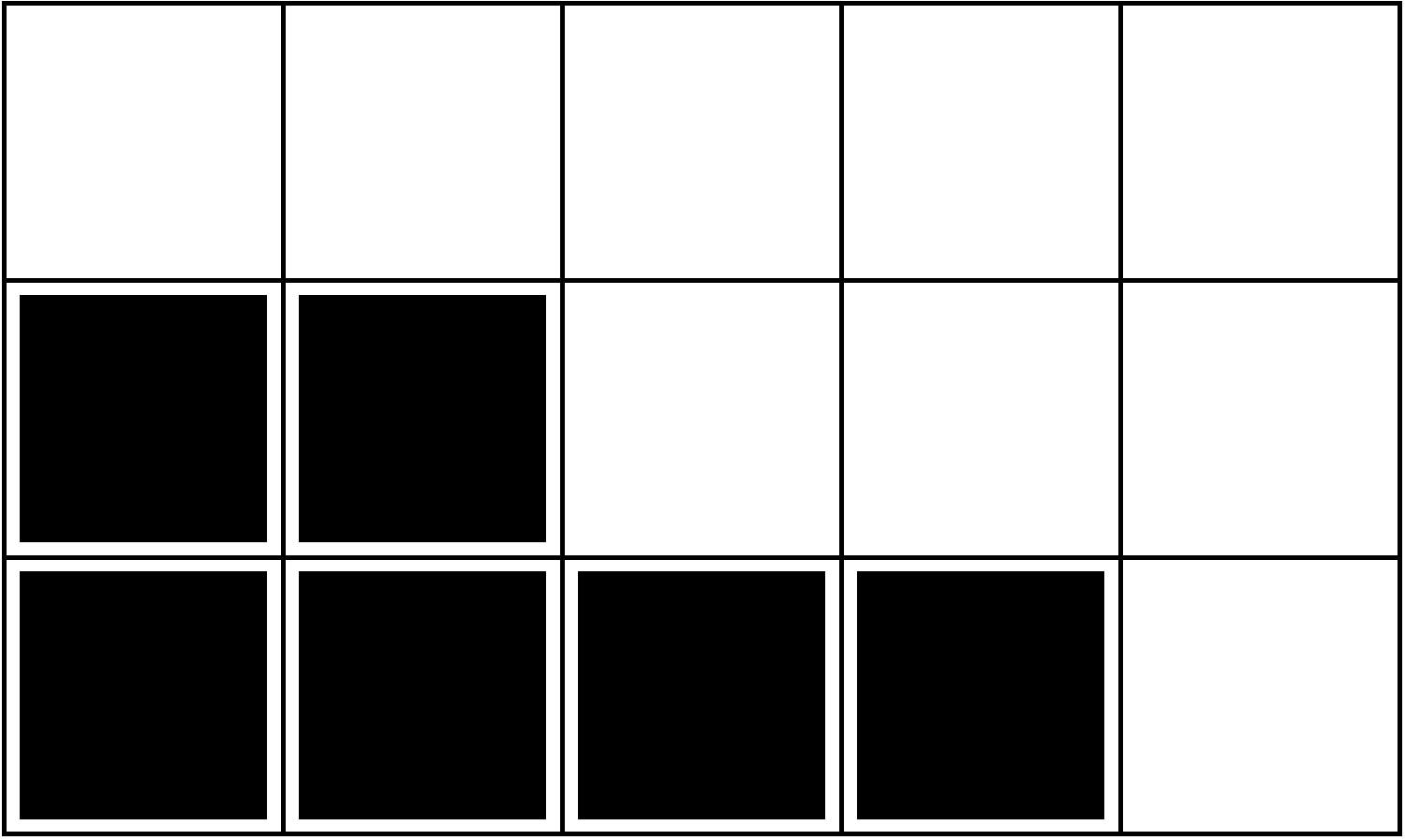}
}
\hfill
\subfloat[]{
\includegraphics[width=0.45\columnwidth]{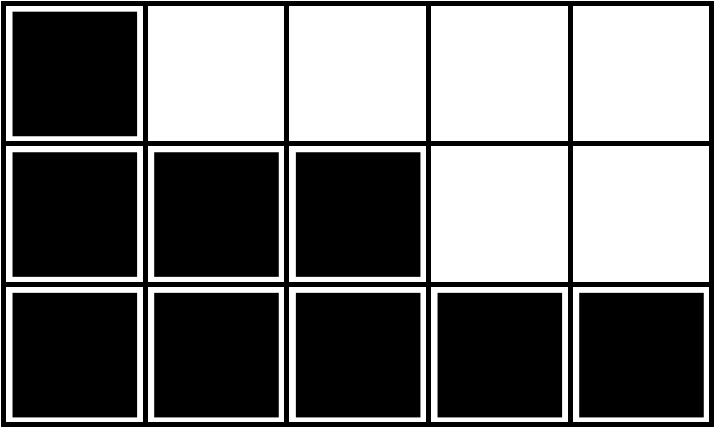}
}

%\\
%\centering
%\hspace*{\fill}
\subfloat[]{
\includegraphics[width=0.27\columnwidth]{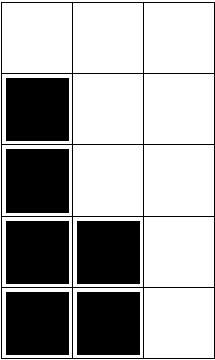}
}
\hspace{0.22\columnwidth}
\subfloat[]{
\includegraphics[width=0.27\columnwidth]{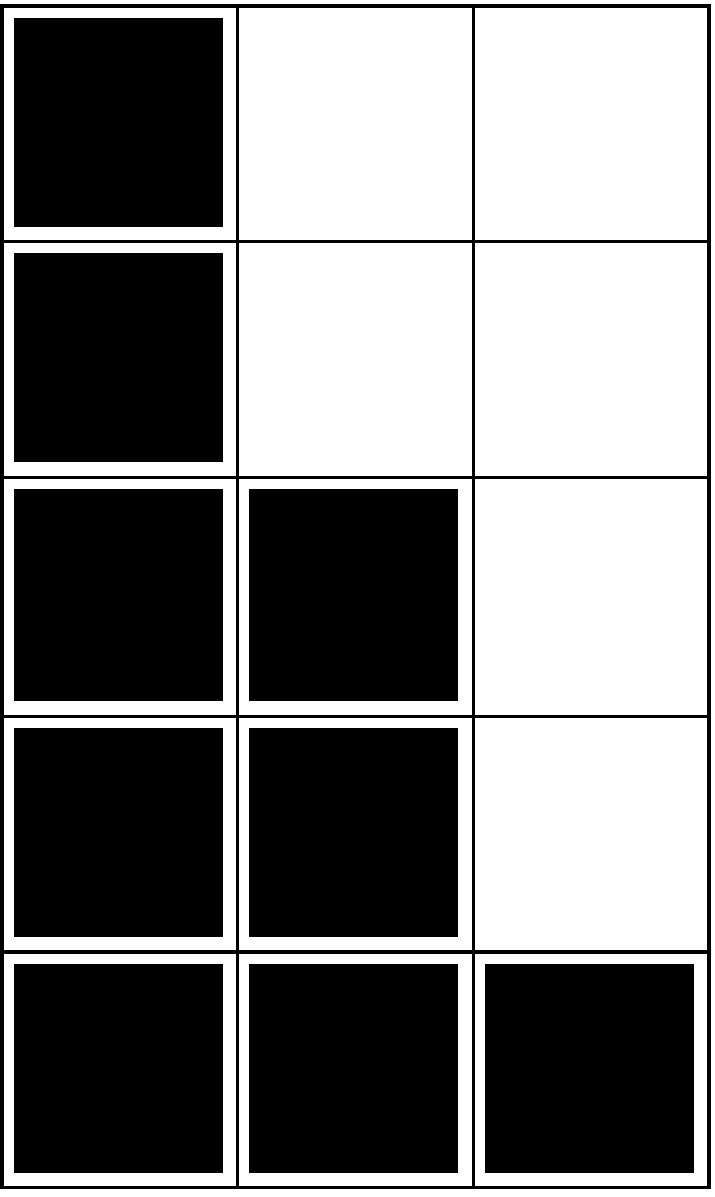}
} %\null
%\hspace*{\fill}

\caption{Young diagrams illustrating 1-1 correspondences. The left-right correspondence is between partitions of $P_{N,M,L}$ and $P_{N,M,NM-L}$, explaining the symmetry of the number of states displayed in FIG. \ref{fig:basissize}.  The top-bottom correspondence is between $P_{N,M,L}$ and $P_{M,N,L}$ giving bijections $B_{Z,L}\leftrightarrow B_{W,L}$ and $D_{Z,L}\leftrightarrow D_{W,L}$.  \label{fig:reflectionrotation}}

\end{figure}

%%%%%%%%%%%%%%%%%
%%%%%%%%%%%%%%%%%

\section{Structure of the lowest-lying wave functions}
\label{sec:rwf}

Using the simple CF basis sets presented in the previous section, we can now diagonalize the interaction Hamiltonian in the simple CF subspace, to produce approximations to the lowest-lying wave functions of the two-component rotating gas. Because the size of the simple CF basis is so much smaller than the size of the Hilbert space, this is often a huge computational simplification. We now take advantage of this to study some aspects of these low-lying states.

In order to increase our physical understanding of the structure of these states, one option is to study density and pair correlation functions. Another approach that is particularly suitable to visualizing vortex structure is to compute the so-called restricted wave function (RWF) $\psi_r (\mathbf{r})$\cite{rwf-1, rwf-2}. To find the RWF of a given many-body wave function $\Psi$, one first calculates a set of particle coordinates $\{\mathbf{r}_i^*\}_{i=1}^{N+M}$ that maximises $|\Psi|^2$. The restricted wave function is defined as
\be{}
\psi_r (\mathbf{r}) = \frac{\Psi(\mathbf{r}, \mathbf{r}_2^*, \mathbf{r}_3^*,\ldots)}{\Psi(\mathbf{r}_1^*, \mathbf{r}_2^*, \mathbf{r}_3^*,\ldots)}
\ee

We see that this function varies as one of the particles is allowed to move from its maximizing position. The amplitude of $\psi_r$ gives the relative amplitude of the many-body wave function compared to the maximum, and the argument gives the change in phase. The vortices can be identified from plots of $\psi_r$ where the nodes of the amplitude meet lines where the phase jumps.

\begin{figure}
\includegraphics[width=\columnwidth]{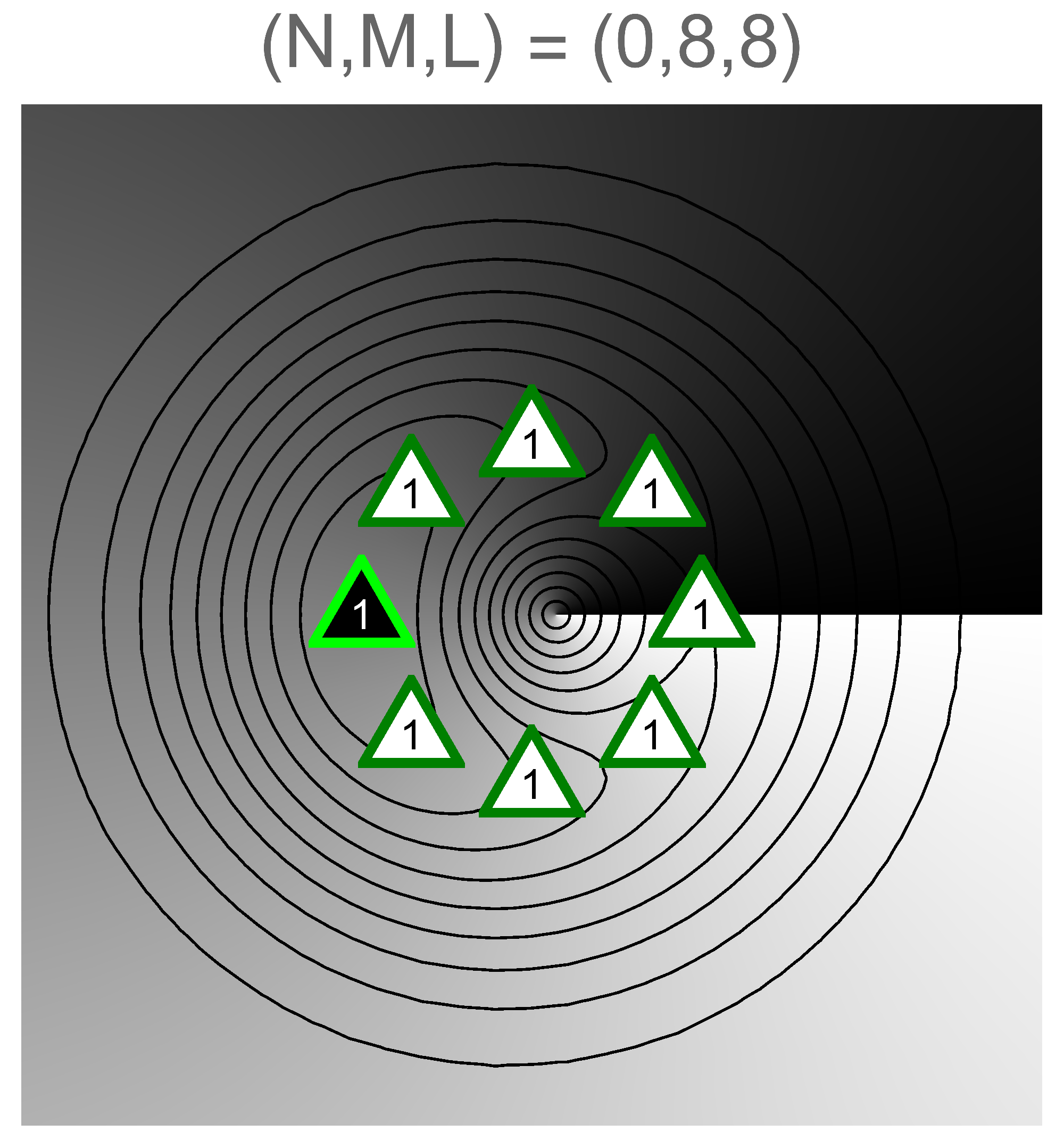}
\caption{Plot of the restricted wave function for the single vortex state of 8 particles. \label{fig:single-vortex}}
\end{figure}

An example of an RWF plot is shown in FIG. \ref{fig:single-vortex}. The wave function $\Psi$ in this example is the exact ground state for a single-species gas with 8 particles at $L=M=8$. For a single component, the cases $L=M$ are known as single vortex states \cite{korslund06, vortex-review}. The triangular plot markers show the optimal positions $\{\mathbf{r}_i^*\}$. The number on each plot marker specifies how many particles share that position. The plot marker with black filling corresponds to the particle whose position $\mathbf{r}$ is varied in the plot. The contour lines show lines of constant amplitude of $\psi_r$ and the color shows the phase change, where black corresponds to $-\pi$ and white to $\pi$. In this case, the configuration of highest $|\Psi|^2$ is a ring of particles, with a vortex clearly visible close to the center of the ring.

In our two-component case, we can define an RWF for each species, $\psi_{rZ}$ and $\psi_{rW}$ for minority and majority species respectively. The difference is simply the species of the particle whose position we vary. We will use triangles pointing down as plot markers for the minority component particles, and triangles pointing up for majority particles (like $\nabla$ and $\Delta$, respectively).

\begin{figure*}

\subfloat[]{
\includegraphics[width=0.32\textwidth]{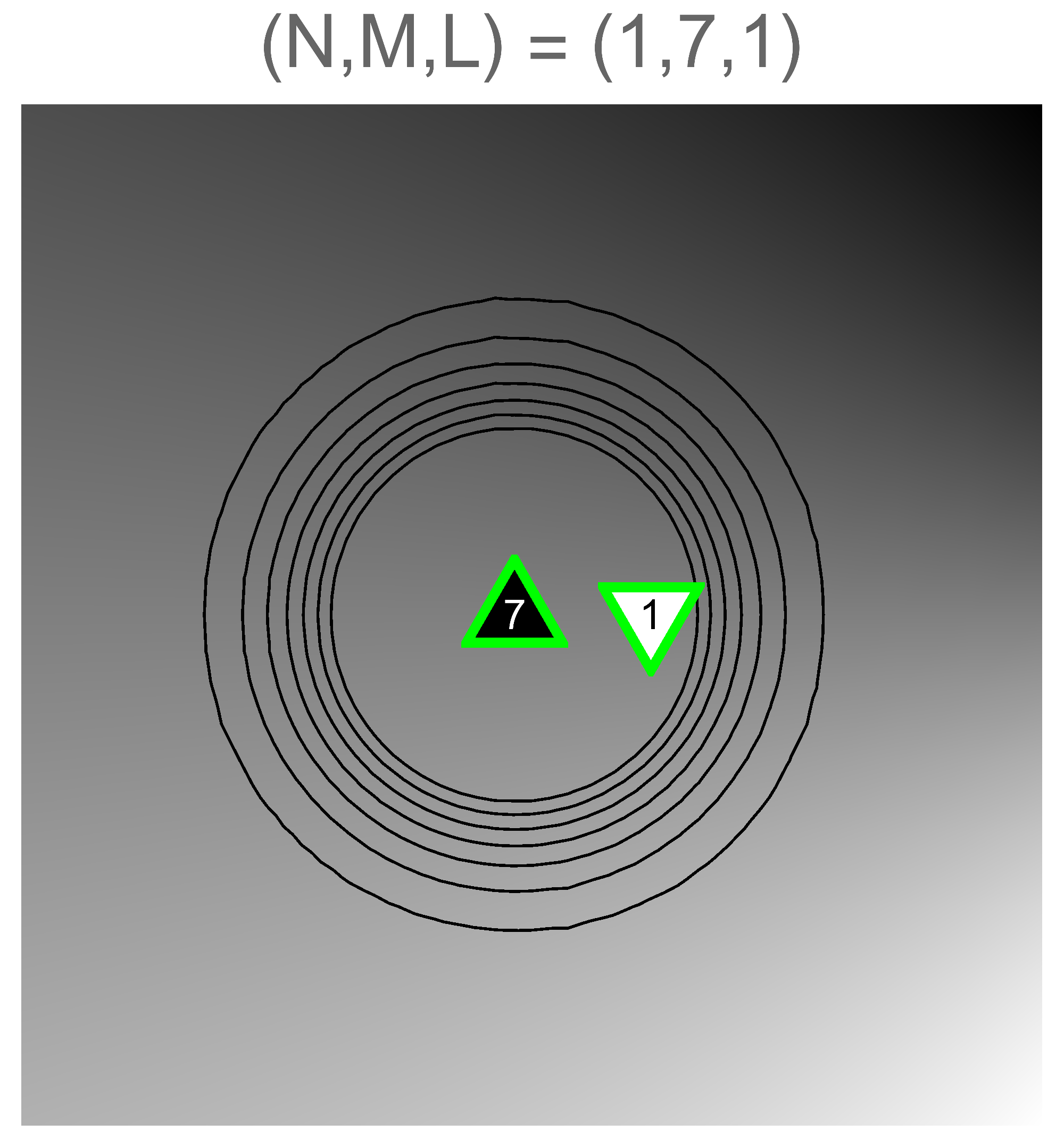}
}
\hfill
\subfloat[]{
\includegraphics[width=0.32\textwidth]{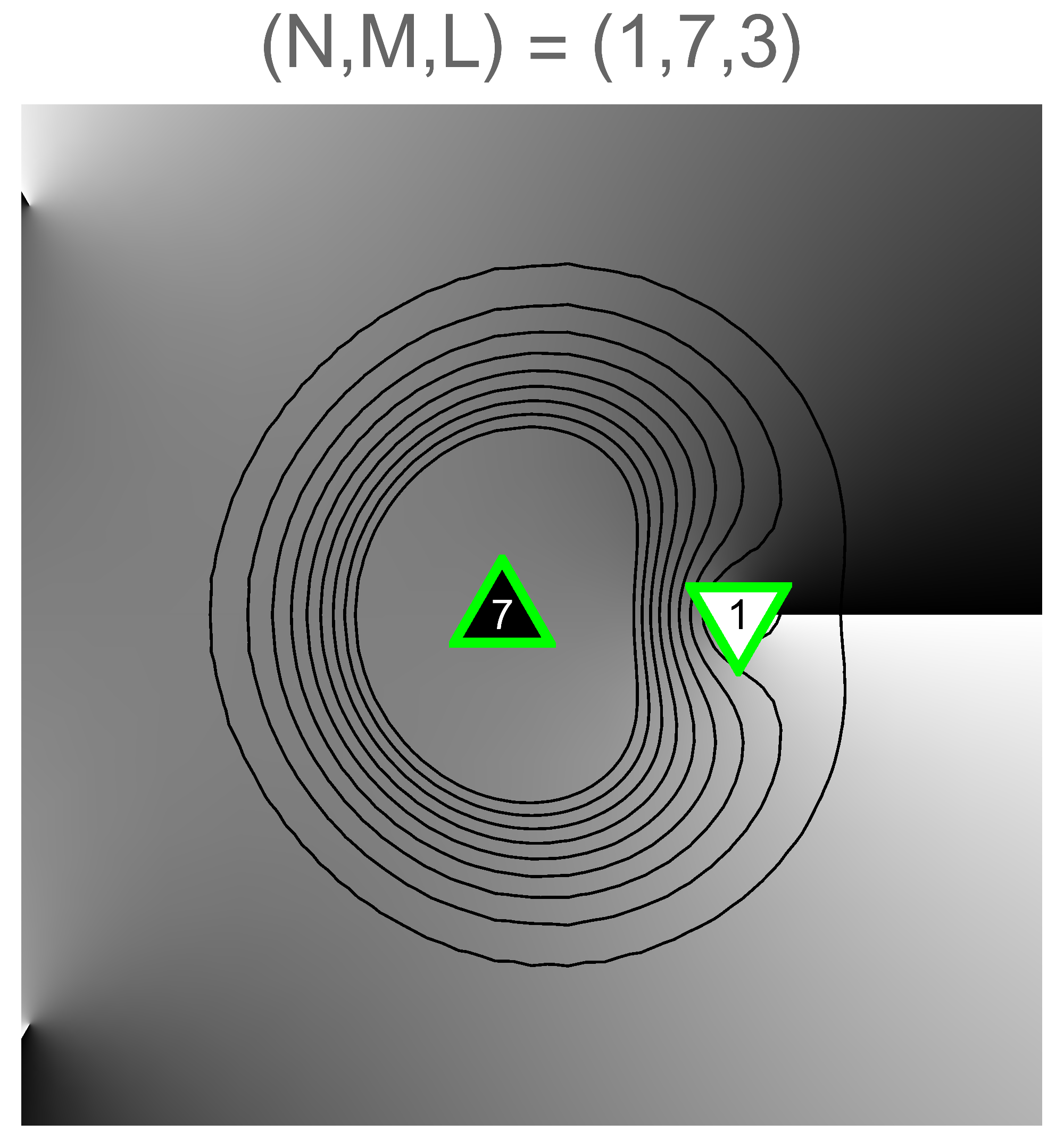}
}
\hfill
\subfloat[]{
\includegraphics[width=0.32\textwidth]{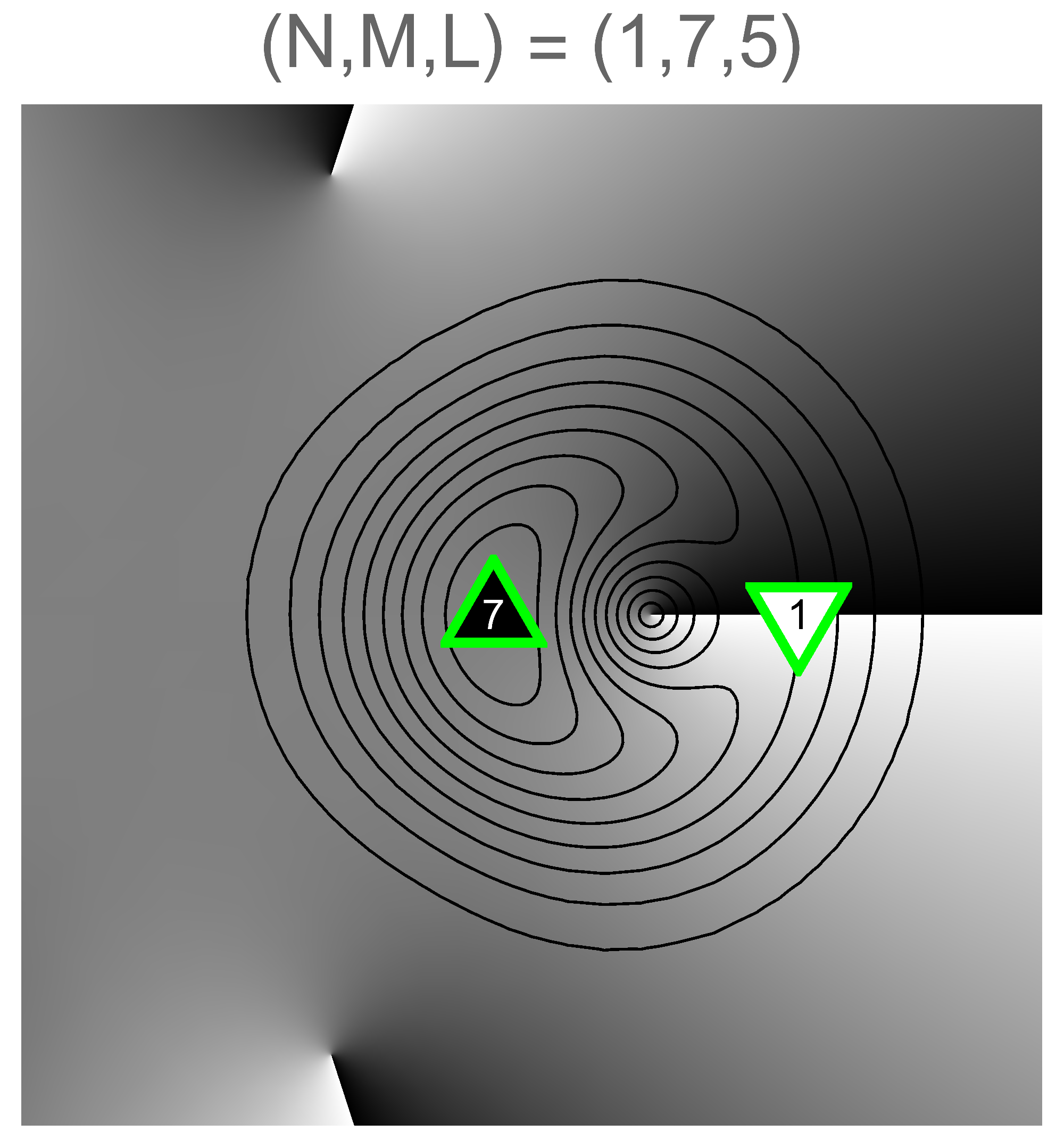}
}

%\\

\subfloat[]{
\includegraphics[width=0.32\textwidth]{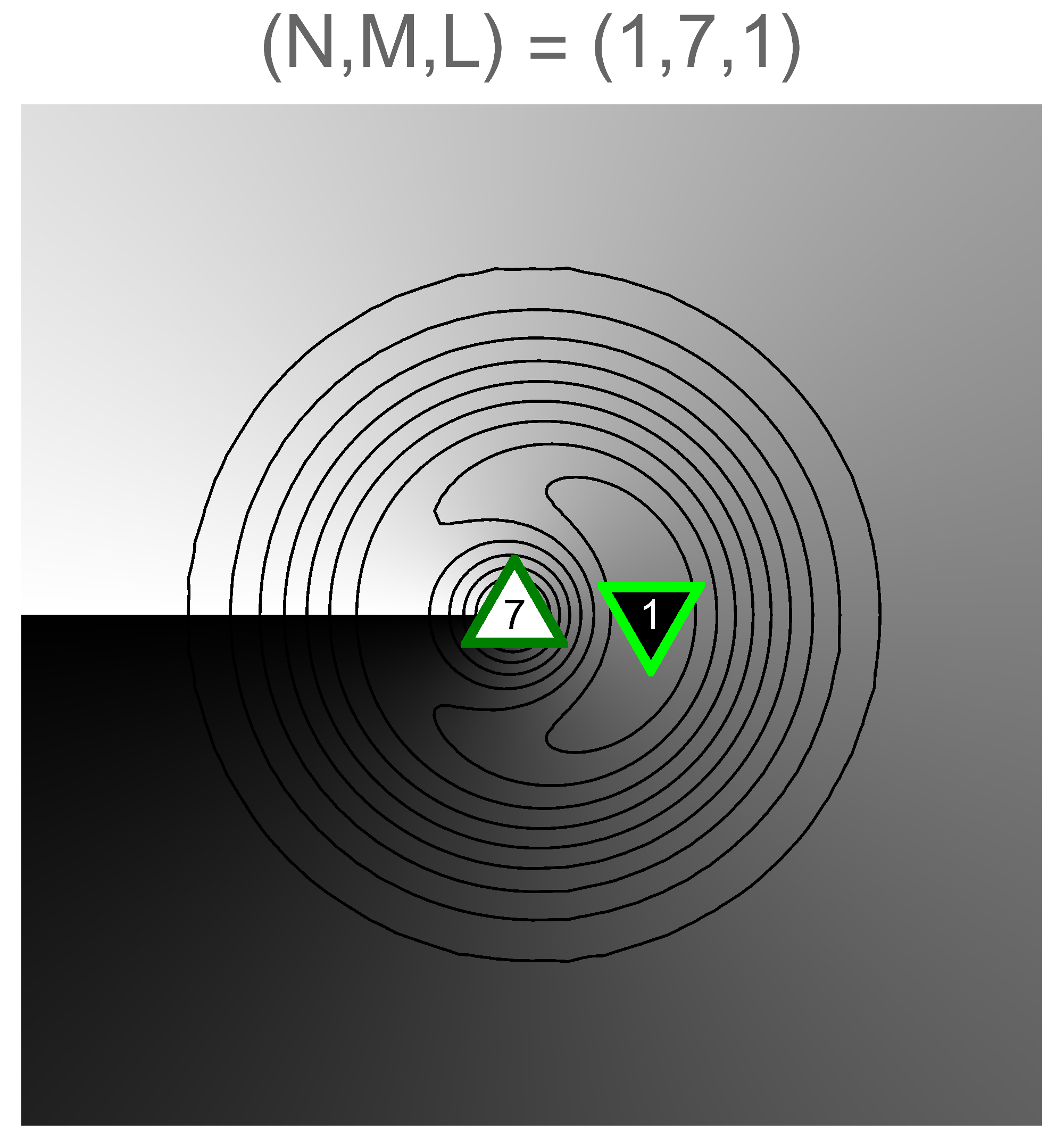}
}
\hfill
\subfloat[]{
\includegraphics[width=0.32\textwidth]{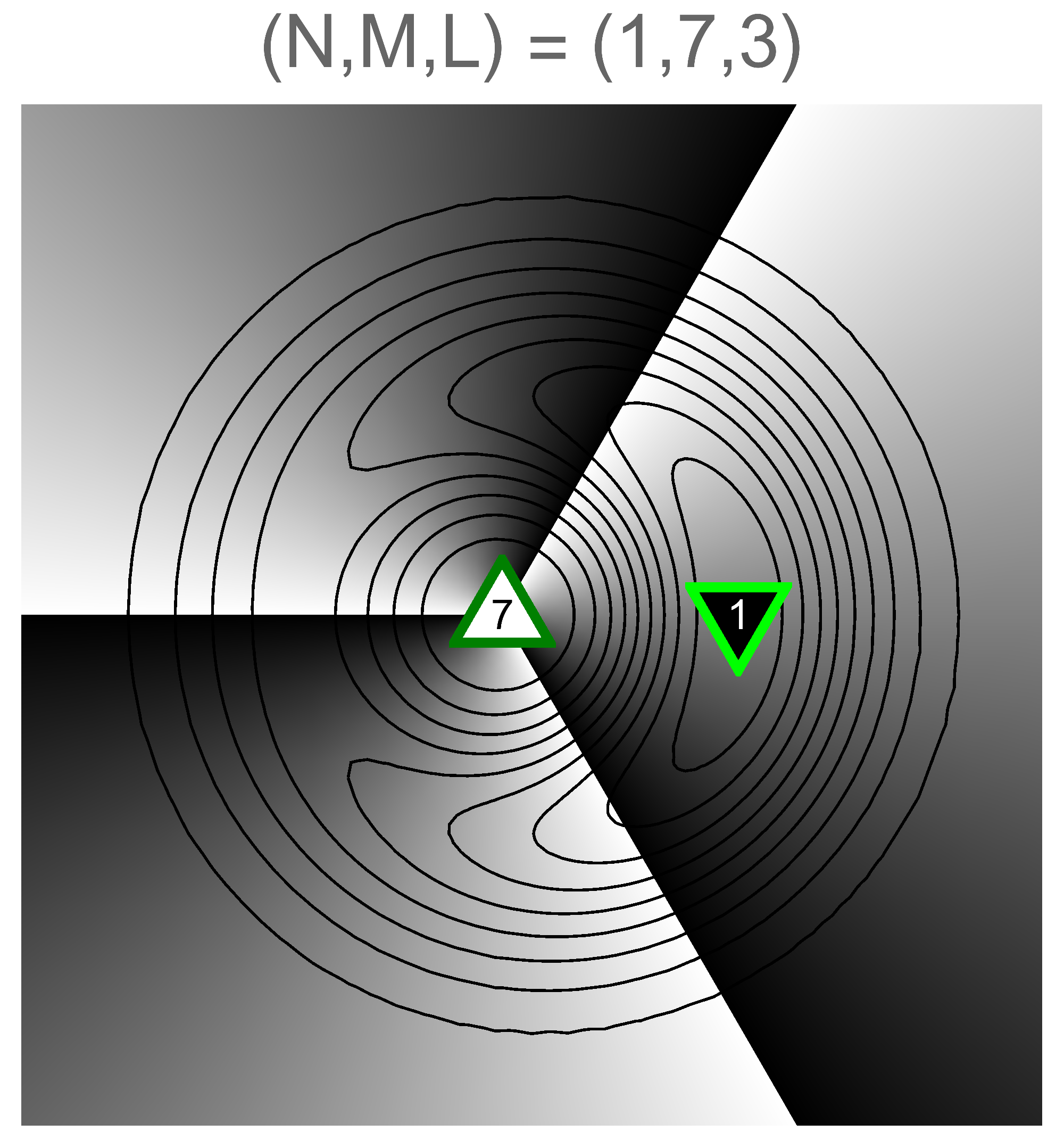}
}
\hfill
\subfloat[]{
\includegraphics[width=0.32\textwidth]{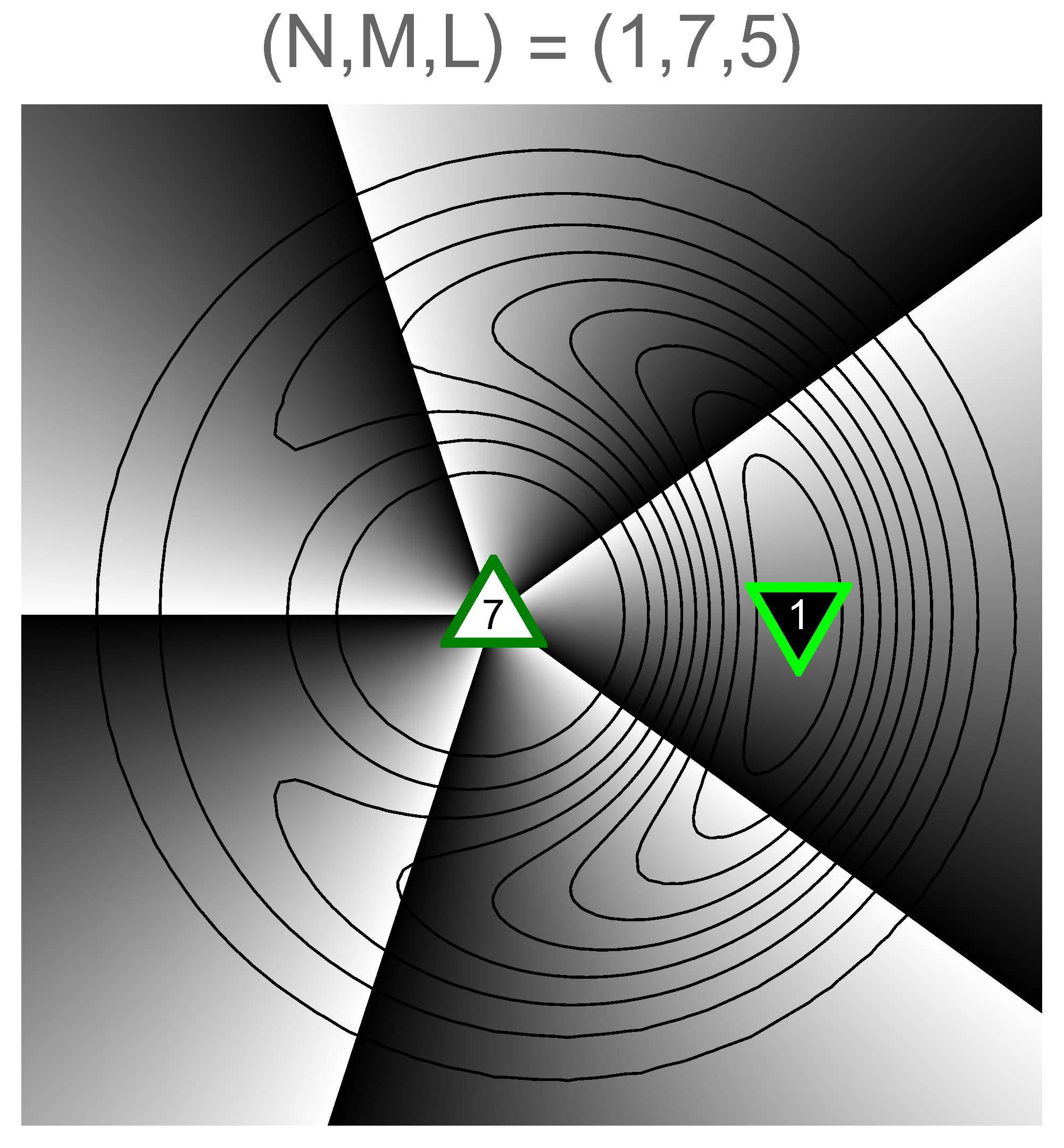}
}

\caption{Restricted wave functions for very low angular momenta for 1 minority and 7 majority particles. (a) - (c) show $\psi_{rW}$, (d) - (f) show $\psi_{rZ}$. \label{fig:17lowL}}

\end{figure*}

In FIG. \ref{fig:17lowL} we see that there is a substantial difference between the RWF plots $\psi_{rZ}$ and $\psi_{rW}$. In the top row we see a vortex approaching the cloud from the side where the lone minority particle is located, starting far outside the cloud and coming closer and closer to the center as $L$ increases. On the other hand, the lone minority particle sees a vortex coinciding with the lump of majority particles already at $L=1$, a so-called coreless vortex \cite{kasamatsu-tsubota-ueda05, christensson08}. This is consistent with findings using full diagonalization \cite{bargi07}. The multiplicity of this vortex increases with $L$, and the vortex and majority particles move together, away from the lone minority particle. This displays an example of how the ``perspectives'' of the two species can be very different for a given state $\Psi$. A similar pattern is also observed for $(N,M) = (2,6), (3,5)$ for $L<M$.

\begin{figure*}

\subfloat[]{
\includegraphics[width=0.32\textwidth]{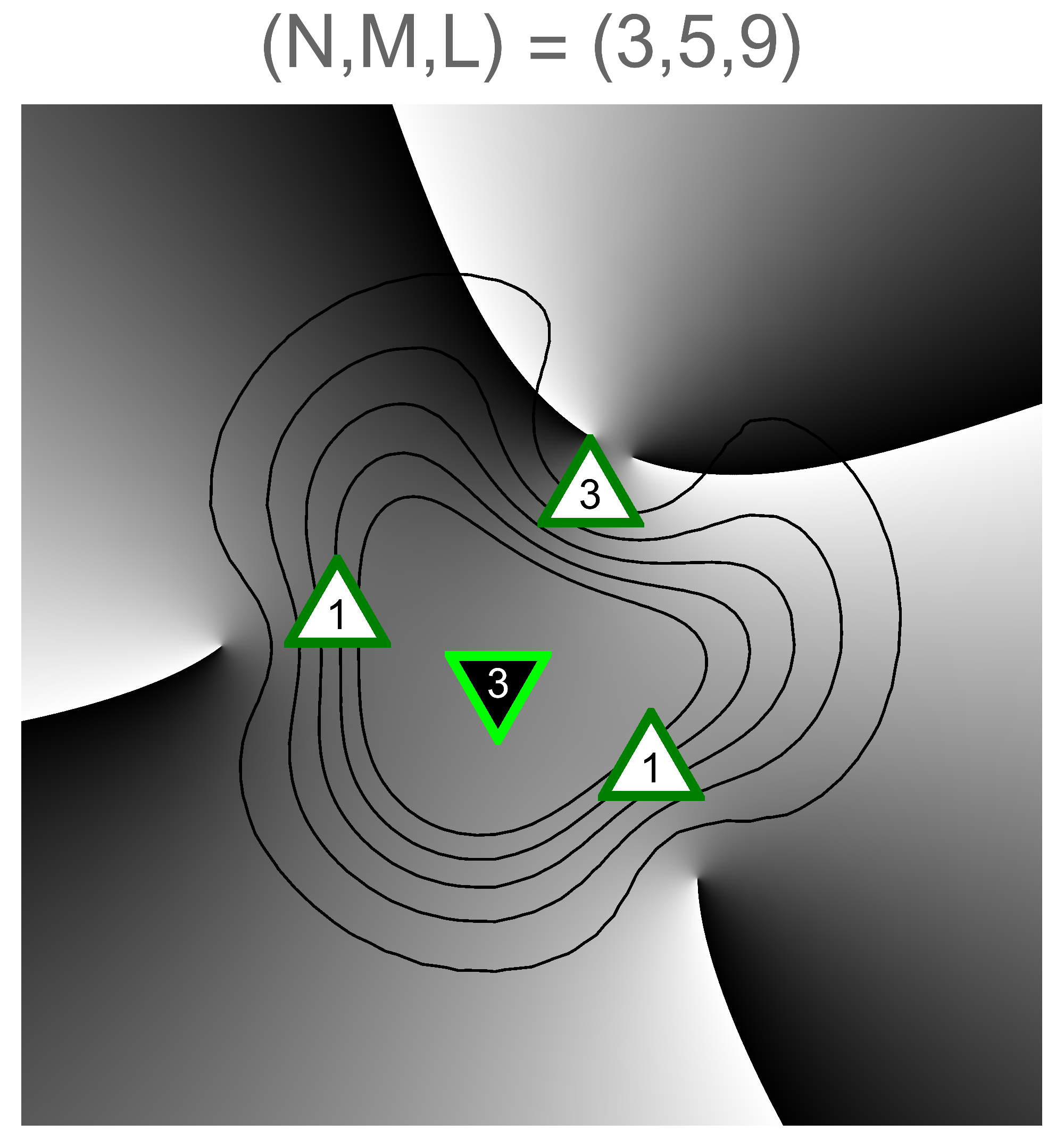}
}
\hfill
\subfloat[]{
\includegraphics[width=0.32\textwidth]{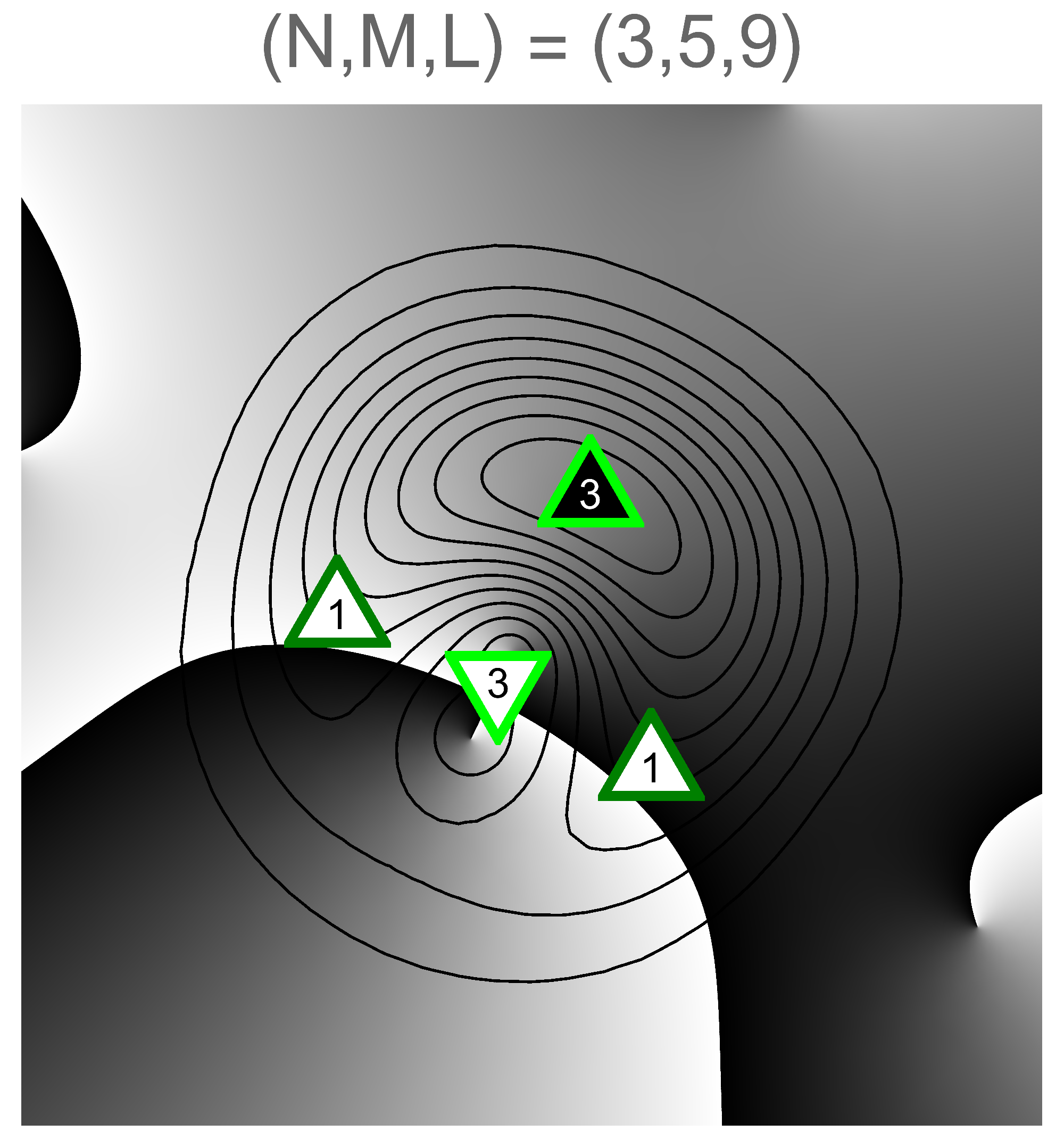}
}
\hfill
\subfloat[]{
\includegraphics[width=0.32\textwidth]{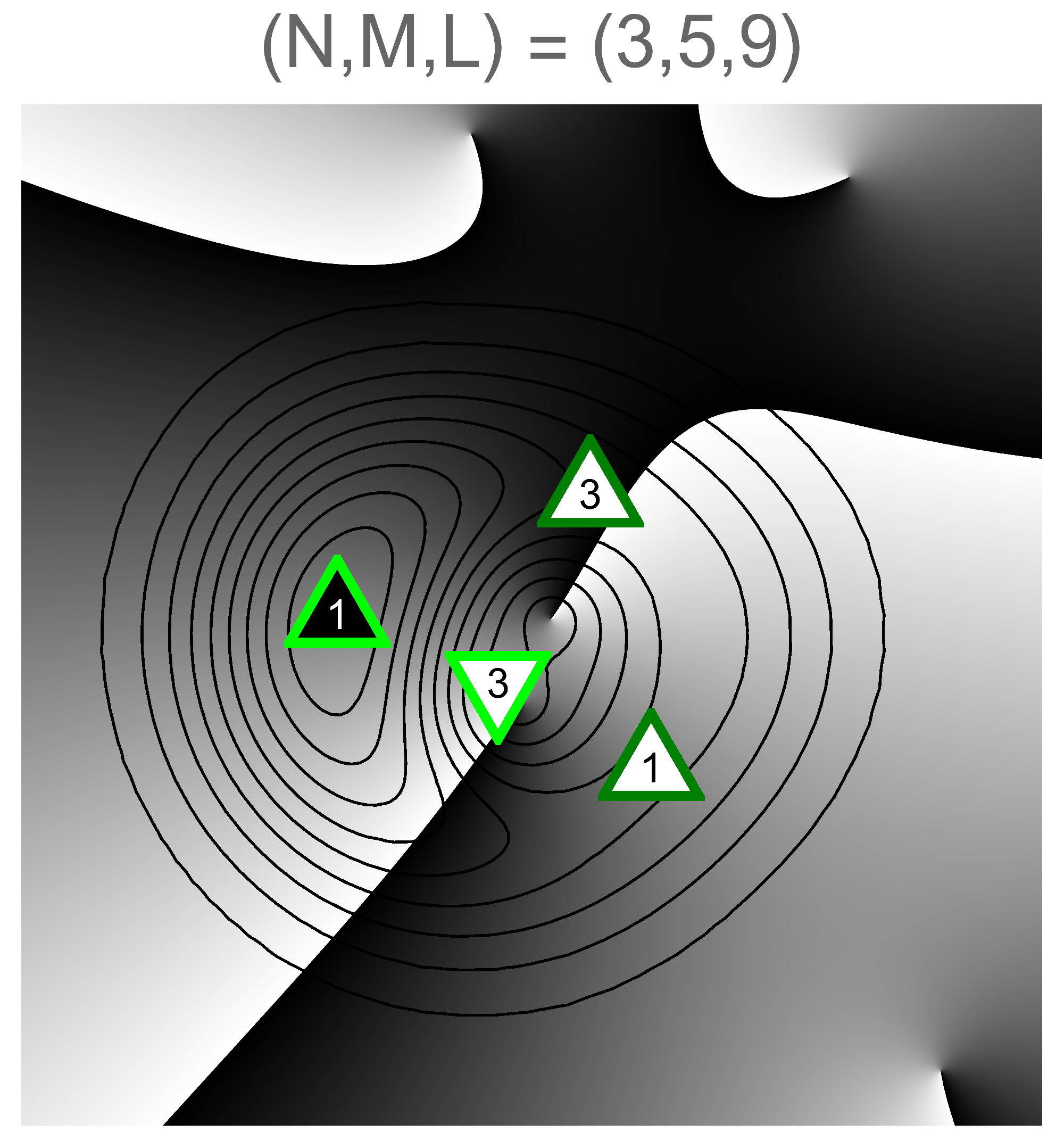}
}

\caption{Three visualizations of the RWF for $(N, M, L) = (3, 5, 9)$. (a) shows $\psi_{rZ}$, while (b) and (c) show $\psi_{rW}$ for two choices of the particle whose position we vary. \label{fig:35}}

\end{figure*}

For $L \geq M$, a variety of configurations are realized as those maximizing $|\Psi|^2$. The common features are that, as $L$ increases and the particles spread out more, particles of the same species still tend to stay close to or on top of some of the other ones, rather than all of them spreading out. The vortices are exclusively located near the particles of the opposite species: no same-species nodes are observed for the $L$ we consider. Finally, the mean number of phase jump lines per particle of the opposite species increases with $L$. A more or less typical situation is shown in FIG. \ref{fig:35}. The majority species has split into three groups, but three of them are still close together. The minority particles are all in the same position. In both (b) and (c), there are two nodes close to the lump of three minority particles. Since the only difference between (b) and (c) is which majority particle coordinate we choose to vary, this demonstrates that the vortex configuration that is seen is not sensitive to that choice, which is what we would expect of a physical vortex state.

\begin{figure*}

\subfloat[]{
\includegraphics[width=0.32\textwidth]{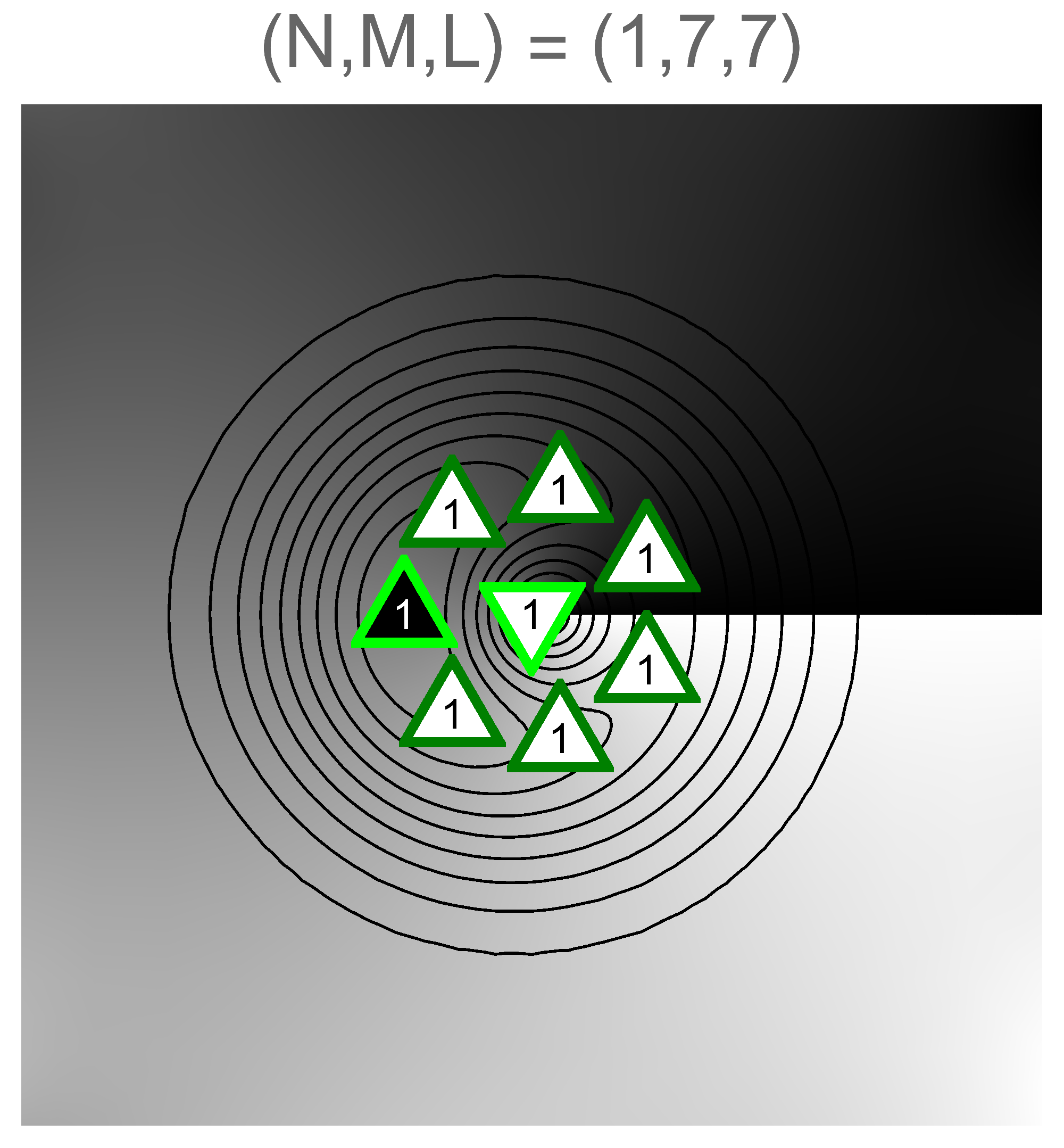}
}
\hfill
\subfloat[]{
\includegraphics[width=0.32\textwidth]{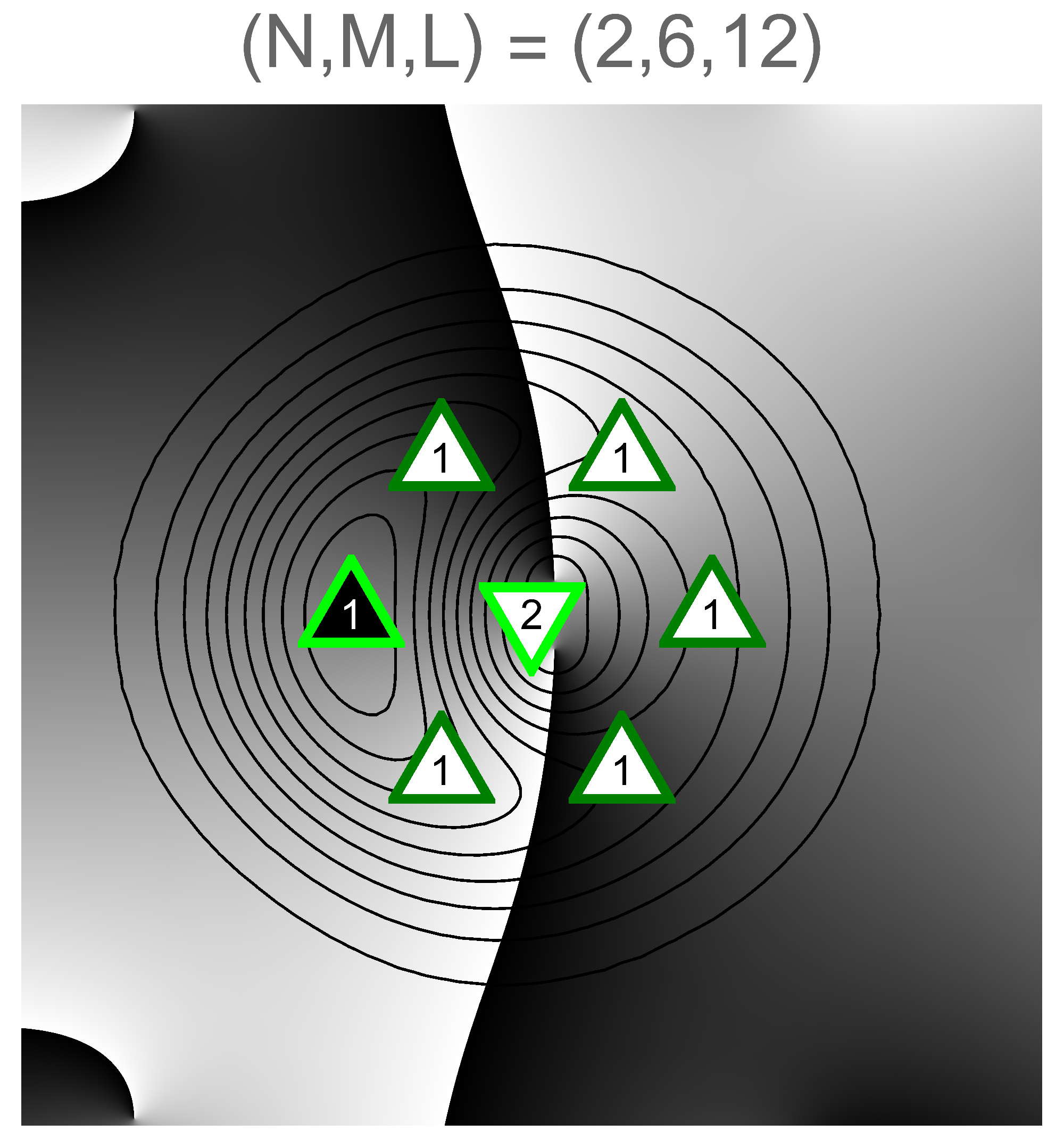}
}
\hfill
\subfloat[]{
\includegraphics[width=0.32\textwidth]{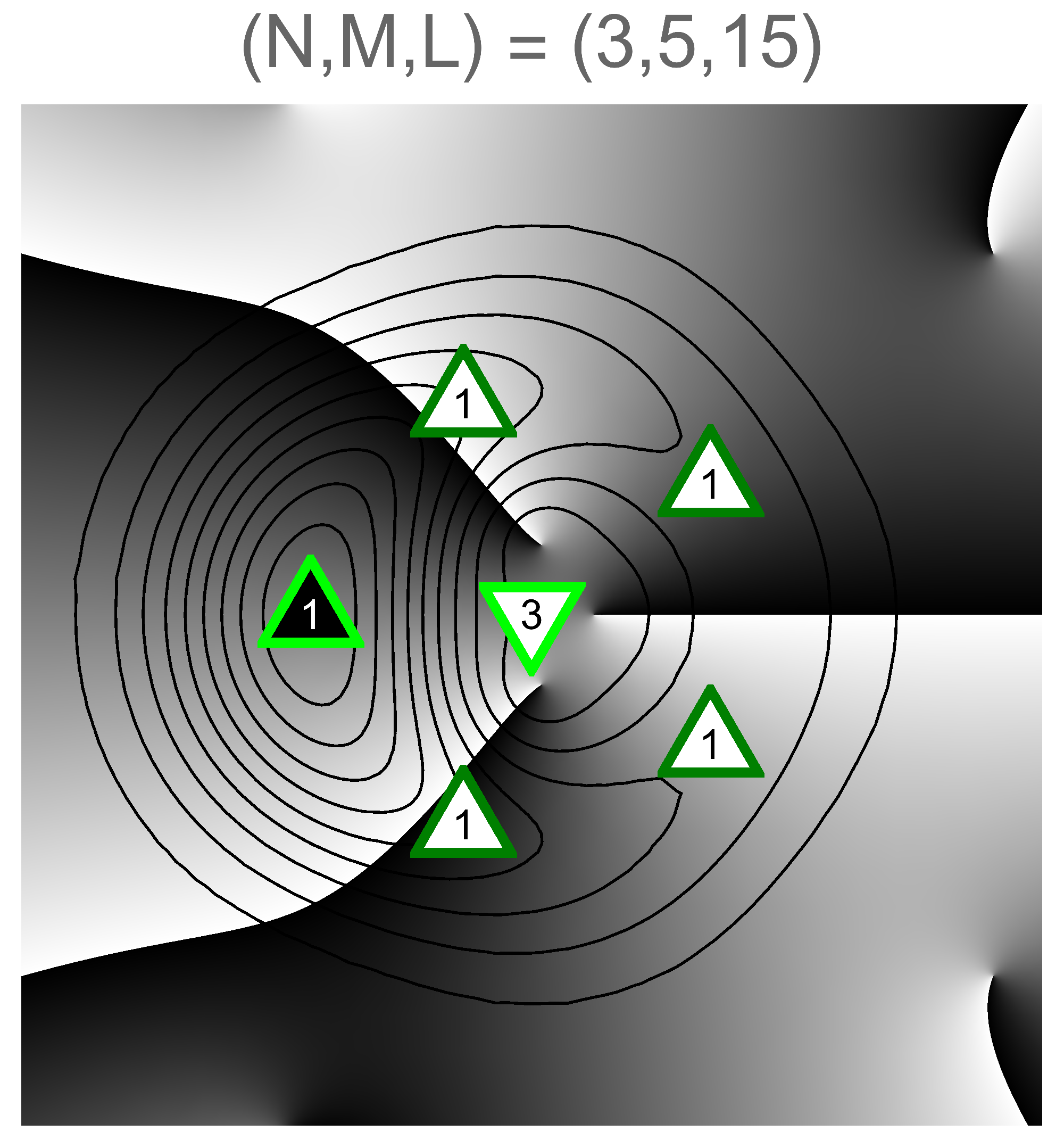}
}

%\\

\subfloat[]{
\includegraphics[width=0.32\textwidth]{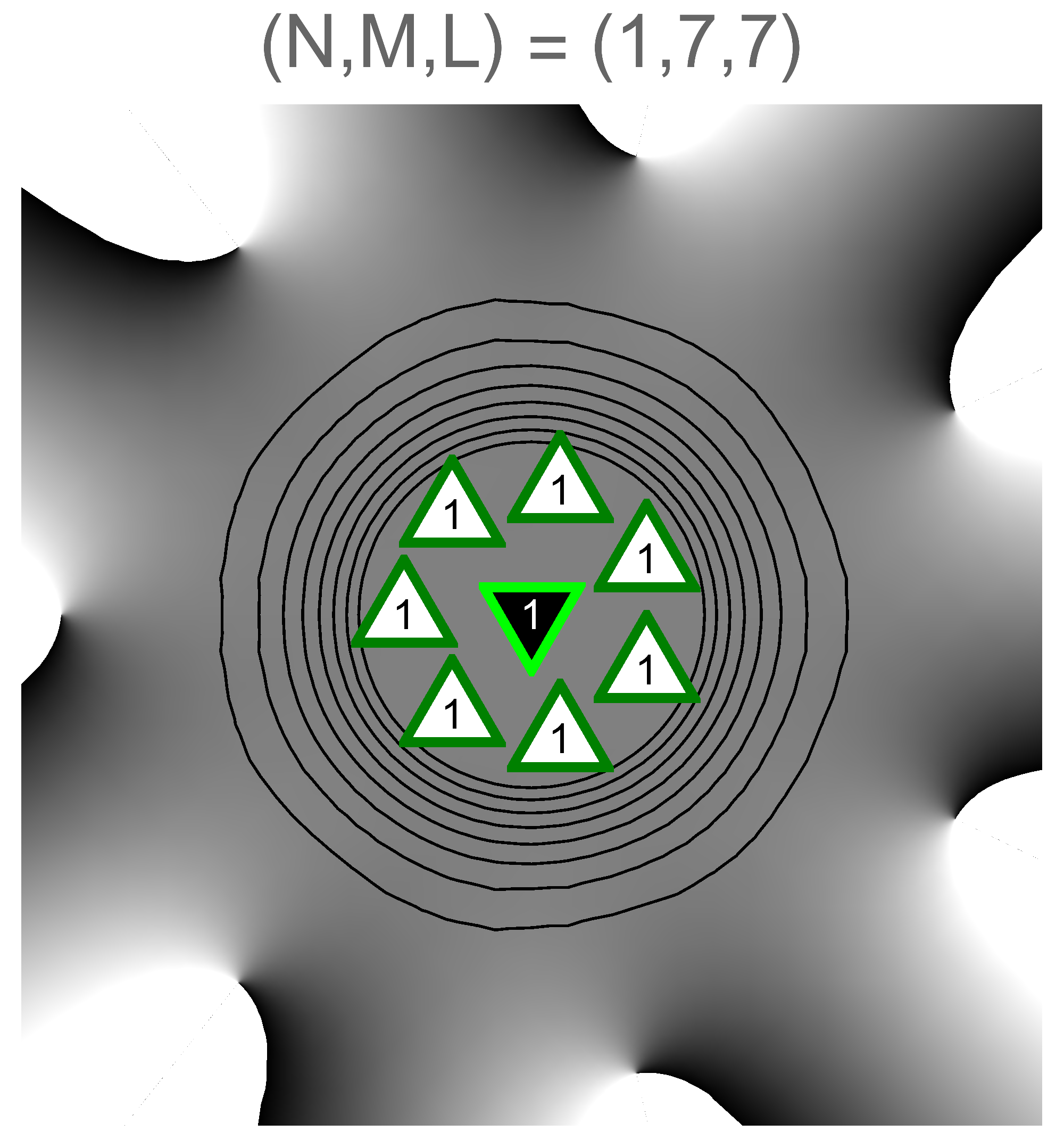}
}
\hfill
\subfloat[]{
\includegraphics[width=0.32\textwidth]{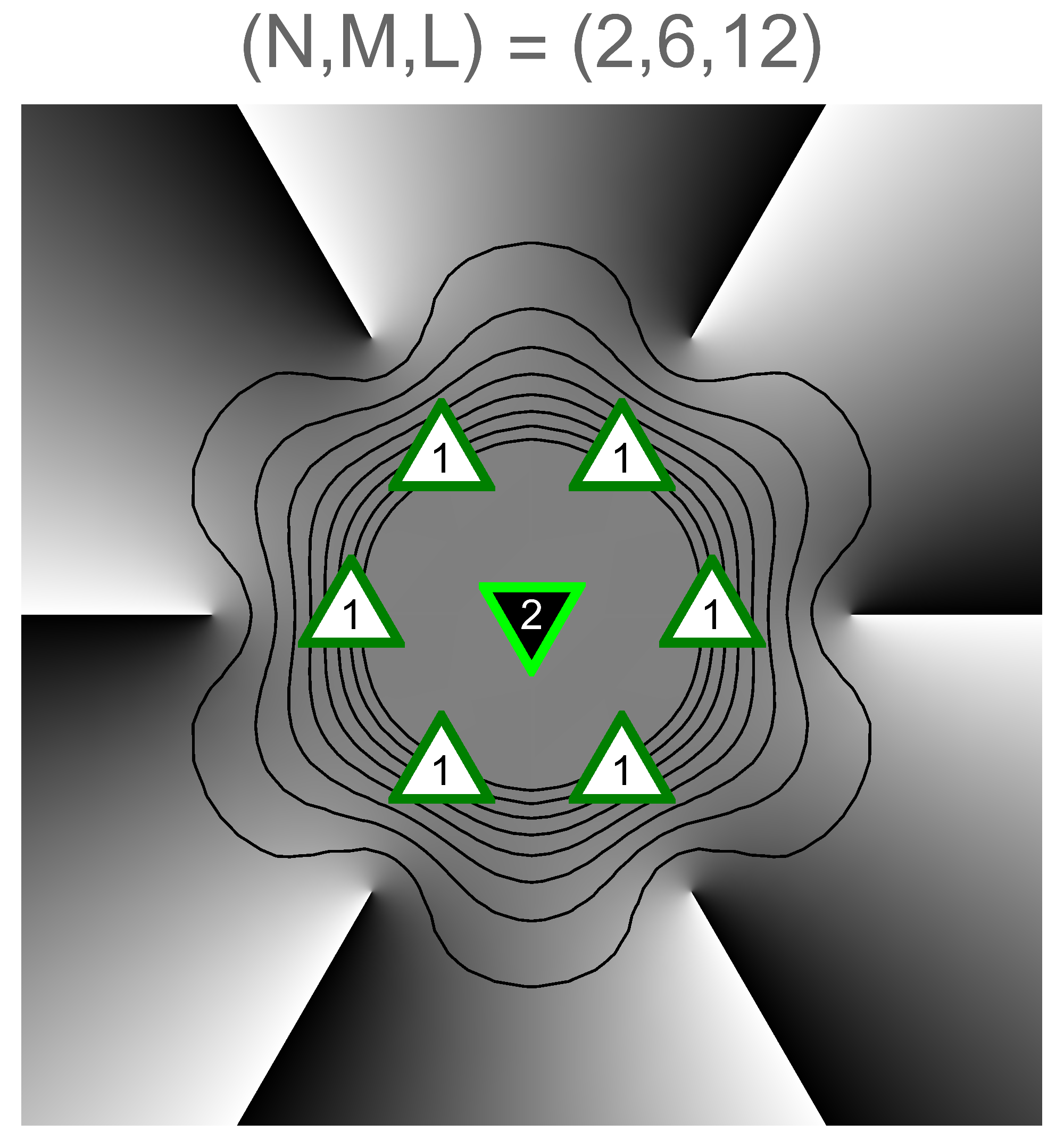}
}
\hfill
\subfloat[]{
\includegraphics[width=0.32\textwidth]{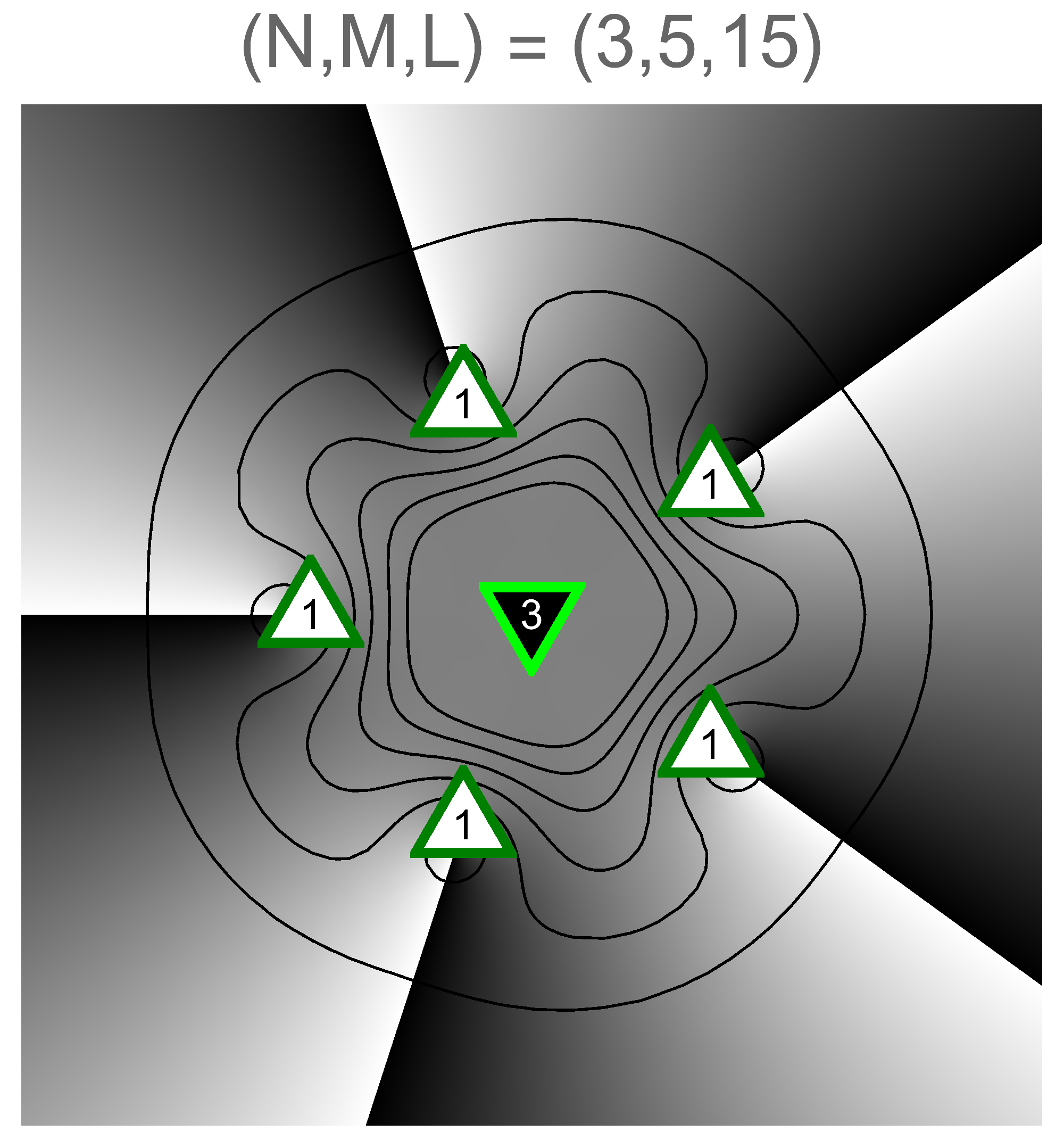}
}

\caption{Restricted wave functions for $L=M N$. (a) - (c) show $\psi_{rW}$, (d) - (f) show $\psi_{rZ}$. \label{fig:Lmax}}

\end{figure*}

Finally we discuss the maximally rotating simple states, namely the states in Eq. (\ref{eq:Lmax}). For a given $N,M$ this is the only possible simple state at $L=M N$, i.e. no CF diagonalization is necessary. The restricted wave functions of some $L=M N$ states are shown in FIG. \ref{fig:Lmax}. Again we see a very clear distinction between minority and majority species particles. The majority particles are positioned on the vertices of a regular $M$-gon, with all the minority particles in the center. In (a) - (c) we see what looks like a single, double and triple vortex structure, respectively, but filled by the minority component. The qualitative behaviour of the amplitude contours is largely the same in the three plots. In particular, \ref{fig:Lmax}(a) looks remarkably similar to FIG. \ref{fig:single-vortex} except for the minority particle in the middle. We will come back to this point in SEC. \ref{sec:overlaps}.

It can be noted that, while the single vortex for a scalar gas appears at $L=N$, the general multiply quantized vortex of winding $k$ appears at $L_k/N < k$; for 20 particles, the double and triple vortices appear at $L/N$ 1.8 and 2.85 respectively \cite{vortex-review}. The results presented here seem to indicate that the vortex of winding $N$ appears at $L = M N$, but that the vortex is a coreless vortex in the majority component.
On the other hand, (d) - (f) show how $\psi_{rZ}$ evolves from a more or less Gaussian distribution in (d) to a situation where one node radially approaches each majority particle from outside the cloud.

%%%%%%%%%%%%%%%%%
%%%%%%%%%%%%%%%%%

\section{Overlaps}
\label{sec:overlaps}

When working with trial wave functions like CF wave functions, especially in a context for which CF was not originally intended, like the slowly rotating Bose systems we are discussing, one should carefully compare the results obtained with ones obtained by other means. In the lowest Landau level, we are fortunate because the Hilbert space of each $L$ sector is \emph{finite}. Given enough computer resources, one can therefore in principle compute the exact spectrum (at least to machine precision) for a given $L$, and compare the CF results to this. In practice, desktop computers can handle two-component systems with up to a total of around 15-20 particles for the $L$ considered in this paper. This method has been used to verify the applicability of the CF construction to scalar \cite{cooper-wilkin99, korslund06} and two-component systems both below and in the quantum Hall regime \cite{grass14, jain13, meyer14}.

In particular, \cite{korslund06, viefers10} showed that the overlap between the exact and CF state for the scalar case $N=L$ (the single vortex) increases to unity in the $N \rightarrow \infty$ limit. As mentioned in the previous section, this is the state plotted in FIG. \ref{fig:single-vortex} for $M=8$. It strikingly resembles the restricted wave function plot of the $L_{\text{max}}$ state with a single minority particle, FIG. \ref{fig:Lmax}. This resemblance, and the fact that the $L_{\text{max}}$ state is unique for given $N,M$, led us to compute the overlap between the exact lowest-lying state and the CF $L_{\text{max}}$ state as function of $M$ for given $N=1,2,3$.

\begin{figure*}

\subfloat[]{
\includegraphics[width=0.32\textwidth]{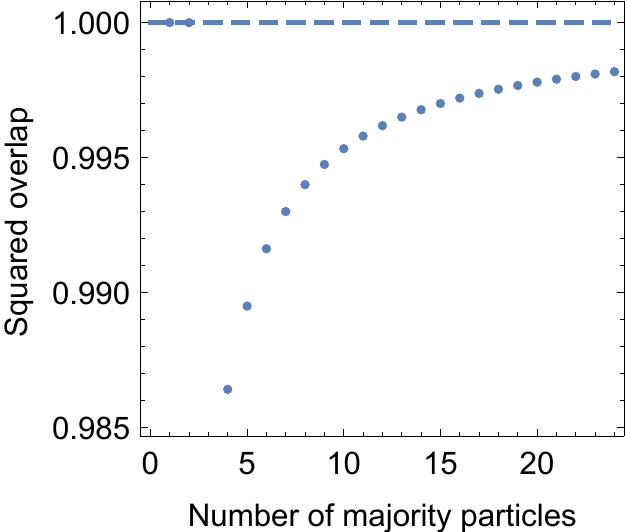}
}
\hfill
\subfloat[]{
\includegraphics[width=0.32\textwidth]{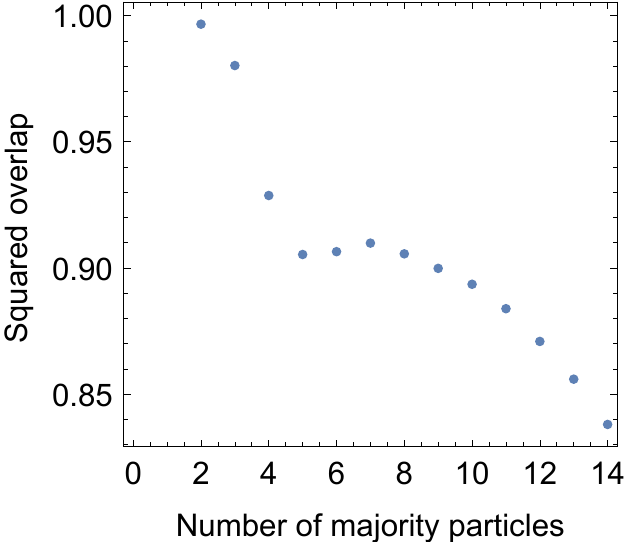}
}
\hfill
\subfloat[]{
\includegraphics[width=0.32\textwidth]{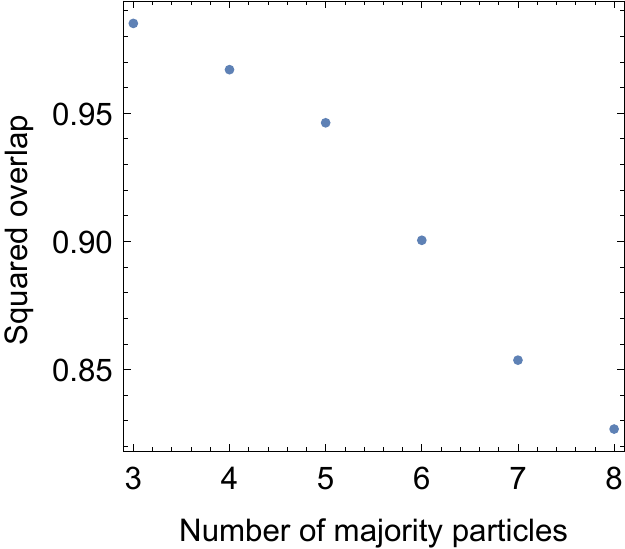}
}

\caption{The squared overlap $|\Braket{\Psi_{\text{CF}} | \Psi_{\text{exact}}}|^2$ between the numerically exact lowest-lying states, and the CF states (\ref{eq:Lmax}), at $L=M N$, for (a) $N=1$, (b) $N=2$, and (c) $N=3$, as function of $M$. \label{fig:overlaps}}

\end{figure*}

In FIG. \ref{fig:overlaps}, the squared overlaps $|\Braket{\Psi_{\text{CF}} | \Psi_{\text{exact}}}|^2$ between the maximally rotating simple CF wave functions at $L=M N$ and the exact lowest-lying states are plottet as functions of $M$ for $N = 1,2,3$ in (a) - (c) respectively. In (a), we see exactly the same convergence to unity as was reported in \cite{korslund06, viefers10} for one component. The CF candidate in the one-component case is \emph{not} simple (simple states don't exist in the scalar case), but it is unique and minimizes the CF cyclotron energy. The exact ground state of the single vortex, on the other hand, is known, and its polynomial part is simply proportional to
\be
\Psi_{\text{s.v.}} = e_M (\tilde{z})
\ee
where $\tilde{z}_i = z_i - 1/M \sum_j z_j$ are the particle coordinates relative to the center of mass, and $e_k$ is the elementary symmetric polynomial of degree $k$.
In fact, for $N=1$, the exact lowest-lying $L=M$ state is known also in the two-component case: it was given in \cite{pap-rei-kavo12} and its polynomial part is proportional to
\be
\Psi_{N=1,L=M} = \tilde{z} e_{M-1}(\tilde{w})  - M e_M(\tilde{w})
\ee
Here $\tilde{\eta}_i = \eta_i - 1/(N+M) \sum_i \eta_i$ are the particle coordinates relative to the center of mass for all particles.

In (b) and (c) however, we clearly see that the squared overlap decreases with system size, as is usually the case with most trial wave function approaches. It should be stressed that the dimension of the relevant sector of Hilbert space grows very rapidly with system size, so given that $\Psi_{L=M N}$ is a unique state in this space, it is still surprising that the overlap with the true lowest-lying state is as large as it is. We should also remember that we have restricted ourselves to simple states in this analysis: the fact that the overlap decreases with system size therefore tells us that higher bands of CF cyclotron energy $E_c$ contribute significantly for larger systems.
%%%%%%%%%%%%%%%%%
%%%%%%%%%%%%%%%%%

\section{Conclusions and outlook}
\label{sec:concl}
 
The main results of this paper are the identification of basis sets of simple CF states,
and the proof that these sets are in fact spanning the simple state subspace and are linearly independent.
We have used these basis sets to revisit the spatial structure of particles and vortices
for the rotating two-component Bose gas in the LLL at low angular momenta, and have paid special attention to the
unique simple CF candidates at angular momentum $L = M N$. We find that for $N=1$ minority particle,
the system mimics the CF candidate for the single vortex state of the scalar rotating gas. This includes an
overlap with the exact wave function that converges to 1 in the $N \rightarrow \infty$ limit. For $N>1$, we observe a coreless vortex of winding $N$ in the majority species at exactly $L = M N$. From the plots in FIG. \ref{fig:overlaps}, however, we see that the exact wave function must have contributions from CF states in higher $E_c$ bands. To answer whether or not this makes any qualitative differences from the results presented here would require further investigation.

As presented in \cite{meyer-lia16}, there are still linear dependence relations between CF candidates in higher $E_c$ bands that are not understood. However, now that we have a good understanding of the relations for simple states, it might be possible to make further progress for these so-called ``compact states'' \cite{jainbook}. They are the relevant candidates for $L > M N$ in the system studied here, but they are also relevant for electrons in strong magnetic fields, confined to quantum dots \cite{jain-kawamura95}.

For other projection methods \cite{jainkam97a} and/or geometries \cite{haldane83}, we expect some linear dependence relations similar to the ones in this paper. The reason is that, as we have seen, it is in fact the Jastrow factor that is responsible for the equations that relate different $\Lambda$-level configuration patterns. Certainly some rules will need to be modified, but the principle of translation invariance should still hold.

%%%%%%%%%%%%%%%%
%%%%%%%%%%%%%%%%

\section*{Acknowledgement}
\noindent
We would like to thank Susanne Viefers for helpful comments on the manuscript. This work was financially supported by the Research Council of Norway.
\

%%%%%%%%%%%%%%%%
%%%%% APPENDIX %%%%%
%%%%%%%%%%%%%%%%

\appendix

\section{Mathematical results}

%%%%%%%%%%%%%%%%
%%%%%%%%%%%%%%%%

\begin{lemma}\label{Lemma1}
There exists $\mathbf{c}\in \mathbb{Q}^{|P_{N,M,\Lambda}|}$ such that 
\be{}
\Phi_Z(\mathbf{a})=\sum_{p\in P_{N,M, \Lambda}} c_p\Delta_Z(p)\Phi_Z(\bm{\alpha}),
\ee
where $\alpha_i=i-1$.  $\Lambda=\sum_{i=1}^N a_i-\alpha_i$.\\
\end{lemma}
\begin{proof}
We use the one to one correspondence between the elements of $P_{N,M, \Lambda}$ and $\Phi_Z(\mathbf{a})$ given by 
\be{}
p\leftrightarrow \Phi_Z(\mathbf{a}(p)),\ \ a_i(p)=i-1+p_{N+1-i}.
\ee
and induce an ordering on the Slater determinants such that
\be{}
p<p'\Leftrightarrow \Phi_Z(\mathbf{a}(p))<\Phi_Z(\mathbf{a}(p')).
\ee
We have that 
\be{}
\Delta_Z(p)\Phi_Z(\bm{\alpha})=\sum_{\sigma\in S_n} \Phi_Z(\bm{\alpha}+\sum_{k=1}^N p_{\sigma_k}\mathbf{e}_k).
\ee
The non-zero terms are all Slater determinants.  The greatest determinant with respect to the ordering occurs when $\sigma$ orders the partition non-decreasingly.  This particular determinant is $\Phi_Z(\bm{a}(p))$ and we may therefore write

\be{}
\Delta_Z(p)\Phi_Z(\bm{\alpha})=\sum_{p'\leq p} n_{pp'}\Phi_Z(\bm{a}(p))
\ee
for some integers $n_{pp'}$.  It is important that 
\be{}
n_{pp}\neq 0.
\ee{}
It is in fact positive, since every permutation $\sigma$ that gives this determinant leaves the elements of $\bm{\alpha}+\sum_{k=1}^N p_{\sigma_k}\mathbf{e}_k$ in increasing order.  It follows that for the smallest partition, $p_{\min}=\min(P_{N,M, \Lambda})$, we have
\be{}\label{lemma2init}
\Phi_Z(\bm{a}(p_{\min}))=\frac{1}{n_{p_{\min}p_{\min}}}\Delta_Z(p)\Phi_Z(\bm{\alpha})
\ee
and in general
\be{}\label{lemma2indu}
\Phi_Z(\bm{a}(p))=\frac{1}{n_{pp}}\Delta_Z(p)\Phi_Z(\bm{\alpha})-\sum_{p'<p} \frac{n_{pp'}}{n_{pp}}\Phi_Z(\bm{a}(p)).
\ee
Eq. (\ref{lemma2init}) says that the Lemma holds for $\Phi_Z(\bm{a})=\Phi_Z(\bm{a}(p_{\min}))$ and eq. (\ref{lemma2indu}) says that it holds for $\Phi_Z(\bm{a})=\Phi_Z(\bm{a}(p))$ if it holds for all $p'<p$.  It must therefore hold for all $p$.
\end{proof}
%%%%%%%%%%%%%%%%
%%%%%%%%%%%%%%%%

\begin{lemma}\label{Lemma2}

For all $p \in P_{N,M, \Lambda}$, there exists $\bm{d}_p\in {\mathbb Q}^{|P_{N,\Lambda,\Lambda}|}$ such that 
\be{}\label{factorlemma}
\Delta_Z(p) =\sum_{\tilde{p}\in P_{N,\Lambda,\Lambda}}d_{p\tilde{p}}\prod_{i=1}^N \Delta_{z^{\tilde{p}_i}},
\ee
where we define $\Delta_{z^0}=N$.
\end{lemma}

\begin{proof}
Our proof is based on induction.  We first define the subset 
\be{}
P_{N,M, \Lambda|K}=\{p\in P_{N,M, \Lambda}\ |\ p_i=0\ \Leftrightarrow\  i>K\}
\ee{}
and note that trivially, for all $p\in P_{N,\Lambda,\Lambda | 1}\supset P_{N,M, \Lambda | 1}$, there exists ${\bf{d}}_p\in\mathbb{Q}^{|P_{N,\Lambda,\Lambda}|}$
such that eq. (\ref{factorlemma}) holds.  There is only one such p and $\Delta_Z(p)=\Delta_{z^\Lambda}$.  Now, we assume as induction hypothesis that (\ref{factorlemma}) holds for all $p\in P_{N,\Lambda,\Lambda | k}$ when $k\leq K$. 

Let $p\in P_{N,\Lambda,\Lambda | K+1}$ and consider the polynomial
\be{}
\prod_{i=1}^N \Delta_{z^{p_i}}.
\ee{}
It is a product of $N$ factors. Each factor is a sum of $N$ differentiation operators with respect to $N$ different variables.  
The resulting terms with the highest number of non-zero exponents are those that do not multiply differentiation operators for the same variable from several different factors $\Delta_{z^{p_i}}$.  This is except for the last $N-K-1$ $\Delta_{z^{p_i}}$ as $\partial_{z_i}^0=\partial_{z_j}^0$ and any of the $N^{N-K-1}$ combinations gives the same.  This restricted part of $\prod_{i=1}^N \Delta_{z^{p_i}}$ can be described by permutations as
\be{}
N^{N-K-1}\sum_{\rho\in S_{K+1}}\prod_{i=1}^{K+1}\partial_{\rho_i}^{p_i}=\frac{N^{N-K-1}}{(N-K-1)!}\Delta_Z(p),
\ee{}
and this means that 
\be{}
\Delta_Z(p)=\frac{(N-K-1)!}{N^{N-K-1}}\prod_{i=1}^N \Delta_{z^{p_i}}+D
\ee{}
where $D$ is a polynomial of differentiation operators where each term has at most $K$ non-zero exponents.  By the induction hypothesis, these can all be written on the form (\ref{factorlemma}).  This implies that the hypothesis is also true for all $k\leq K+1$ and completes the proof.
\end{proof}
%%%%%%%%%%%%%%%%
%%%%%%%%%%%%%%%%

As was introduced in \cite{meyer-lia16}, the simple CF states obey

\begin{lemma}[Generalized translation invariance]\label{gentrans}

\be{}
(\Delta_{z^n}+\Delta_{w^n})\Psi=0 
\ee
for all integers $n>0$.  
\end{lemma}

Ordinary translation invariance is captured in the case $n=1$. 

\begin{proof}
The operator $(\Delta_{z^n}+\Delta_{w^n})$ commutes with the Slater determinants, so it is enough to show that $(\Delta_{z^n}+\Delta_{w^n})J=0$.  This result is independent of the splitting of particles into $Z$ and $W$ type.  We can therefore use variables $\{\eta_i\}_{i=1}^{N+M}$.  We have that  
\be{}\label{jastrowinvarians} 
\Delta_{\eta^n}J &=\sum_{i=1}^N\partial_{\eta_i}^n\sum_{\rho\in S_{N+M}} (-1)^{|\rho|} \prod_{j=1}^{N+M}\eta_{\rho_j}^{j-1}\\
&=\sum_{i=1}^N\sum_{\rho|\ \rho_i>n} (-1)^{|\rho|} \frac{(\rho_i-1)!}{(\rho_i-n-1)!}\prod_{j=1}^{N+M}\eta_{j}^{\rho_j-1-n\delta_{j,i}}.
\ee
Now, for each pair $(i,\rho)$, there is a unique pair $(i',\rho')$ such that $\rho_i-n=\rho_{i'}$ and $\rho'_{i'}-n=\rho'_i$ and $\rho_j=\rho'_j$ for all $j\neq i,i'$.  Since $\rho$ and $\rho'$ only differ by a permutation of two elements, $(-1)^{|\rho|+|\rho'|}=-1$ and the corresponding terms in eq. (\ref{jastrowinvarians}) cancel out.  Since this happens for all $(i,\rho)$, it follows that $\Delta_{\eta^n}J=0$.
\end{proof}
%%%%%%%%%%%%%%%%
%%%%%%%%%%%%%%%%

\begin{lemma}[Linear independence]\label{linearindependence}%Linear ind

\be{}
B_{Z,L}=\{\Phi_Z(\bm a(p))\Phi_W(\bm\beta)J\ |\ p\in P_{N,M,M N-L}\}
\ee
is a linearly independent set. 
\end{lemma} 

\begin{proof}
If
\be{}
\Psi_{p}=\Phi_Z(\bm a(p))\Phi_W(\bm\beta)J\in B_{Z,L},
\ee
then we denote the projection onto $w_i=0\ \forall w_i$ by $\bar\Psi_p$, and the set of projected states by $\bar B_{Z,L}$.  We will show that this set is linearly independent and therefore $B_{Z,L}$ is as well.

The $w$-independent terms of $\Psi_p$ arise when the $\partial_{w_i}$ operators act on the $M$ lowest order variables in the Jastrow factor.  We can therefore write

\be{}
\bar{\Psi}_p &=(\prod_{k=0}^{M} k!) \sum_{\sigma,\tau\in S_N} (-1)^{NM+|\sigma|+|\tau|}\prod_{i=1}^N\partial_{z_{\sigma_i}}^{a_i(p)}\prod_{j=1}^N z_{\tau_j}^{M-1+j}\\
&=(\prod_{k=0}^{M} k!) \sum_{\sigma,\tau\in S_N} (-1)^{NM+|\sigma|}\prod_{i=1}^N\partial_{z_{\tau_i}}^{a_{\sigma_i}(p)} z_{\tau_i}^{M-1+i}.
\ee
We are interested in the particular symmetrized term that occurs when 
$\sigma$ is the identity operator.  We name this term $t_p$, and it can be written as
\be{}
t_p=K_p\sum_{\tau\in S_N}\prod_{i=1}^N x_{\tau_i}^{M-p_i},
\ee
where $K_p$ is an integer that results from differentiation and possibly a permutation sign.  The term has the property that the smallest exponent is as great as possible among terms in $\bar\Psi_p$. Given that, the second smallest exponent is as great as possible and so on. We use this to define an ordering on the $t_p$ terms, saying that 
\be{}
t_p<t_{p'}
\ee
iff the $k$'th least exponent of $t_p$ is \emph{greater} than the $k$'th least exponent of $t_{p'}$ and their $k-1$ least exponents are pairwise equal.  This is equivalent to $p<p'$, and if we pick another non-zero term $t'_p$ of $\bar\Psi_p$ by choosing a $\sigma$ that is not the identity, then we have $t_p<t'_p$.  We index the partitions such that 
\be{}
p_1<p_2<...<p_{P_{N,M,NM-L}},
\ee
Now, if $i<j$, then $t_{p_i}<t_{p_j}<t'_{p_j}$.  This means that $\bar\Psi_i$ contains a polynomial term that is not contained in any $\bar\Psi_j$ for all $j>i$.  This means that if there exists a linear dependence relation
\be{}
c_1\bar\Psi_{p_1}+...+c_{|P_{N,M,NM-L}|}\bar\Psi_{p_{|P_{N,M,NM-L}|}}=0,
\ee  
then the leftmost non-zero coefficient of this relation must be 0, and the relation must therefore be trivial.
\end{proof}

%%%%%%%%%%%%%%%%
%%% BIBLIOGRAPHY %%%%
%%%%%%%%%%%%%%%%

\end{document}